\documentclass[final,11pt]{article}

\usepackage[dvipsnames]{xcolor}

\def\showkeys{0}

\def\showcolorlinks{1}

\ifnum\showkeys=1
\usepackage[color]{showkeys}
\fi

\ifnum\showcolorlinks=1
\usepackage[
pagebackref,
letterpaper=true,
colorlinks=true,
urlcolor=blue,
linkcolor=blue,
citecolor=blue,
]{hyperref}
\fi

\ifnum\showcolorlinks=0
\usepackage[
pagebackref,
letterpaper=true,
colorlinks=false,
pdfborder={0 0 0}
]{hyperref}
\fi

\usepackage{prettyref}

\usepackage{nicefrac}

\newcommand{\nfrac}{\nicefrac}

\def\FullBox{\hbox{\vrule width 6pt height 6pt depth 0pt}}

\def\qed{\ifmmode\qquad\FullBox\else{\unskip\nobreak\hfil
\penalty50\hskip1em\null\nobreak\hfil\FullBox
\parfillskip=0pt\finalhyphendemerits=0\endgraf}\fi}

\newenvironment{proof}{\begin{trivlist} \item {\bf Proof.~~}}
   {\qed\end{trivlist}}

\usepackage{comment,amsmath,amssymb,amsfonts}
\usepackage{epsfig} \usepackage{latexsym,nicefrac,bbm}
\usepackage{xspace}
\usepackage{color,fancybox,graphicx,url,subfigure,fullpage}
\usepackage[top=1in, bottom=1in, left=1in, right=1in]{geometry}
\usepackage{tabularx} \usepackage{hyperref}
\usepackage{mathptmx}

\usepackage{enumitem}
\usepackage{booktabs}

\newtheorem{theorem}{Theorem}[section]
\newtheorem{lemma}[theorem]{Lemma}

\newtheorem{fct}[theorem]{Fact}
\newtheorem{defn}[theorem]{Definition}

\newtheorem{remark}[theorem]{Remark}
\newtheorem{conjecture}[theorem]{Conjecture}

\newcommand{\E}{{\mathbb E}}
\newcommand{\bmu}{{\mu}}

\newcommand{\calF}{{\cal F}}

\newcommand{\calU}{{\cal U}}
\newcommand{\calP}{{\cal P}}

\newcommand{\Real}{{\mathbb R}}

\newcommand{\x}{x}
\newcommand{\y}{y}
\renewcommand{\v}{v}
\renewcommand{\u}{u}
\newcommand{\z}{z}

\newcommand{\V}{{\mathbf{V}}}

\newcommand{\GLC}{{$(\ell_2^2,\ell_1,O(1))$-Conjecture}}
\renewcommand{\epsilon}{\varepsilon}

\renewcommand{\v}{{v}}
\newcommand{\w}{{w}}
\renewcommand{\u}{{u}}

\newcommand{\PM}{{\{-1,1\}^N}}
\renewcommand{\cal}{\mathcal}

\newcommand{\MC}{{\sc MaxCut}}
\newcommand{\MUC}{{\sc MinUncut}}
\newcommand{\SC}{{\sc SparsestCut}}
\newcommand{\BS}{{\sc BalancedEdge-Separator}}
\newcommand{\UG}{{\sc UniqueGames}}
\newcommand{\VC}{{\sc VertexCover}}
\newcommand{\QP}{{\sc  QuadraticProgramming}}

\newcommand{\defeq}{\stackrel{\textup{def}}{=}}

\begin{document}

\title{\bf The Unique Games Conjecture, Integrality Gap for Cut Problems and
Embeddability of Negative Type Metrics into $\ell_1$\footnote{A preliminary version of this paper appeared in
FOCS 2005, see \cite{KhotV05}.}}

\author{Subhash A. Khot\thanks{Subhash A. Khot. New York University, NY, USA. Email: khot@cims.nyu.edu}  \and Nisheeth K. Vishnoi
\thanks{Microsoft Research, Bangalore, India. Email: nisheeth.vishnoi@gmail.com.}}

\date{}
\maketitle

\begin{abstract}

 In this paper, we disprove a conjecture of Goemans \cite{GoemansSurvey} and Linial \cite{LinialSurvey} (also see \cite{ARV,MatousekBook}); namely,  that
every negative
type metric embeds into $\ell_1$ with constant distortion.
 We show that for an arbitrarily small constant
$\delta
> 0$, for all large enough $n$, there is an $n$-point
negative type metric which requires distortion at least  $(\log\log n)^{\nfrac{1}{6}-\delta}$ to embed into $\ell_1.$

 Surprisingly, our construction is inspired by the Unique Games
Conjecture (UGC) of Khot \cite{KhotUCSP}, establishing a
previously  unsuspected connection between probabilistically checkable proof systems (PCPs)
and the
theory of metric embeddings. We first prove that the {UGC} implies a
super-constant hardness result for the (non-uniform) {\SC} problem.
Though this hardness result relies  on the {UGC}, we
demonstrate, nevertheless, that the corresponding {PCP}
reduction can be used to construct an  ``integrality gap
instance'' for \SC. Towards this, we first
construct an integrality gap instance for a natural {SDP}
relaxation of {\UG}. Then we ``simulate'' the {PCP}
reduction and ``translate'' the integrality gap instance of
{\UG} to an integrality gap instance of {\SC}.
This enables us to prove a $(\log \log n)^{\nfrac{1}{6}-\delta}$
integrality gap for  {\SC}, which is known to be equivalent to
the metric embedding lower bound.

\end{abstract}

\newpage
\tableofcontents

\newpage

\section{Introduction}\label{sec:intro}
\subsection{Metric Embeddings and their Algorithmic Applications}
In recent years, the theory of metric embeddings has played an
increasing role in algorithm design. The best approximation algorithms
for several {NP}-hard problems rely on techniques (and theorems)
used to embed one metric space into another while preserving all pairwise distances up to
a certain {\it not too large} factor, known as the {\it distortion} of the embedding.

\smallskip
Perhaps, the most well-known application of this paradigm is the {\SC} problem. Given an $n$-vertex graph along with a set of {\it demand pairs}, one seeks to find a  non-trivial partition of the graph that
minimizes the {\it sparsity}, i.e., the ratio of the number of edges cut to the number of demand pairs cut. Strictly speaking,
the problem thus defined is the
{\it non-uniform} version of {\SC} and in the absence of a qualification, we  always mean the non-uniform version.
In contrast, the {\it uniform} version refers to
the special case when the set of demand pairs consists of all possible $\binom{n}{2}$ vertex pairs.
In the uniform
version, the sparsity is the same (up to a factor $2$ and a normalization factor of $n$) as the ratio of the number of edges cut to
the size of the smaller side of the partition. A closely related problem is the {\BS} problem where one desires a partition
that cuts a constant fraction of  demand pairs and minimizes the number of the edges cut. In its uniform version,
one desires a {\it balanced} partition, say a $(\nfrac{1}{3}, \nfrac{2}{3})$-partition,\footnote{In the uniform case, for a parameter $b \in (0, \nfrac{1}{2}],$ a partition of the vertex set is said to be a  $(b,1-b)$ partition if each side of the partition contains at least $b$ fraction of the vertices.} that minimizes the number of the edges cut.

\smallskip
Bourgain \cite{BourgainMet} showed that every $n$-point metric
embeds into $\ell_2$  (and, hence, into $\ell_1$ since every $n$-point subset of $\ell_2$ isometrically embeds into
$\ell_1$) with distortion
$O(\log n)$.  Aumann and Rabani \cite{AR} and Linial, London and Rabinovich \cite{LLR}  independently gave a
striking application of Bourgain's theorem: An $O(\log n)$
approximation algorithm for {\SC}.
The approximation ratio is
exactly the distortion incurred in Bourgain's theorem. This gave
an alternate approach to the seminal work of Leighton and Rao
\cite{LR}, who obtained an $O(\log n)$ approximation  algorithm
for {\SC} via a linear programming (LP) relaxation based on multi-commodity flows.\footnote{In fact, algorithms based on
metric embeddings work for the more general {\it non-uniform}
version of {\SC}. The Leighton-Rao
algorithm worked only for the uniform version. }
 It is well-known that an $f(n)$ factor algorithm for {\SC}  can be
 used iteratively  to design an $O(f(n))$ factor  algorithm for
{\BS}. In particular, in the uniform case, given a graph that has a $(\nfrac{1}{2},
\nfrac{1}{2})$-partition cutting an $\alpha$ fraction of the edges,
the algorithm produces a $(\nfrac{1}{3}, \nfrac{2}{3})$-partition
that cuts  at most $O(f(n)\alpha)$ fraction of the edges. Such
partitioning algorithms are very useful as  sub-routines in
the design of graph theoretic algorithms via the divide-and-conquer paradigm.

\smallskip
The  results of \cite{AR,LLR} are based on the {\it metric {LP}
relaxation} of {\SC}. Given an instance $G(V,E)$ of {\SC}, let
$d_G$ be the $n$-point metric obtained as a solution to this {LP}. The metric $d_G$ is then embedded into $\ell_1$ via
Bourgain's  theorem. Since $\ell_1$ metrics are non-negative
linear combinations of {\it cut metrics}, an embedding into $\ell_1$
essentially gives the desired sparse cut (up to an $O(\log n)$
approximation factor). Subsequent to this result, it
was realized that one could write a {semi-definite programming} (SDP) relaxation of
{\SC} with the so-called {\it triangle inequality constraints} and enforce an additional condition that the metric $d_G$
belongs to a special subclass of metrics called the {\it negative
type  metrics} (denoted by  $\ell_2^2$). Clearly, if $\ell_2^2$
embeds into $\ell_1$ with distortion $g(n)$, then one gets a $g(n)$ approximation to {\SC} via
this {SDP} (and in particular the same upper bound on the {\it integrality gap} of the {SDP}).

\smallskip
The results of  \cite{AR,LLR} led  to the
conjecture that $\ell_2^2$  embeds into $\ell_1$
with distortion $C,$ where $C$ is an absolute constant. This conjecture has been attributed to Goemans \cite{GoemansSurvey} and Linial \cite{LinialSurvey}, see  \cite{ARV, MatousekBook}.
This conjecture, which we  henceforth refer to as  the \GLC, if
true, would have had tremendous algorithmic applications (apart
from being an important mathematical result). Several problems,
specifically cut problems (see \cite{DLBook}),  can be formulated
as optimization problems over the class of $\ell_1$ metrics, and
optimization over $\ell_1$ is an {NP}-hard problem in general.
However, one can optimize over $\ell_2^2$  metrics in
polynomial time via {SDP}s (and since $\ell_1 \subseteq \ell_2^2$, this is indeed a relaxation).
Hence, if $\ell_2^2$ metrics were embeddable into $\ell_1$ with
constant  distortion, one would get a computationally efficient constant factor
approximation to $\ell_1$ metrics.

\smallskip
However, no better embedding of $\ell_2^2$   into $\ell_1,$
other  than Bourgain's $O(\log n)$ embedding (that works for all
metrics), was known. A breakthrough result of
Arora, Rao and Vazirani ({ARV}) \cite{ARV} gave an
$O(\sqrt{\log n})$  approximation to
 (uniform) {\SC} by showing that the integrality gap of the
{SDP} relaxation is $O(\sqrt{\log n})$ (see also \cite{NRS}
for an alternate perspective on {ARV}).  Subsequently, {ARV} techniques were used by Chawla, Gupta and R\"acke \cite{CGR}
to give an $O(\log^{\nfrac{3}{4}} n)$ distortion embedding of $\ell_2^2$
metrics into $\ell_2$  and, hence, into $\ell_1.$ This result was
further improved to $O(\sqrt{\log n} ~\log\log n)$ by Arora, Lee
and Naor \cite{ALN}.\footnote{This implies, in particular, that
every $n$-point $\ell_1$ metric embeds into $\ell_2$ with
distortion $O(\sqrt{\log n} ~\log\log n)$, almost matching decades
old $\Omega(\sqrt{\log n})$ lower bound due to Enflo \cite{Enflo}.}
Techniques from {ARV} have also been applied to obtain
an $O(\sqrt{\log n})$ approximation to {\MUC} and  related problems
\cite{ACMM},  to {\sc VertexSeparator} \cite{FHL}, and to obtain
a  $2-O(\nfrac{1}{\sqrt{\log n}})$ approximation to {\VC} \cite{Karakostas}. It was conjectured in the {ARV}
paper that the integrality gap of the {SDP} relaxation of (uniform)
{\SC}   is bounded from above by an absolute constant.\footnote{The {\GLC}
implies the same also for the non-uniform version.} Thus, if the {\GLC} and/or the {ARV}-Conjecture were true, one would
potentially get a constant factor approximation to a host of
problems, and perhaps, an algorithm for {\VC} with an
approximation factor better than $2.$

\subsection{Our Contribution} \label{sec:our-results}
The main contribution of this paper is the disproval of the \GLC.
This is an immediate corollary of the following theorem  which proves the existence of an  appropriate integrality gap instance for non-uniform {\BS}. See Section \ref{sec:prelim} for a formal description of the \GLC,
the non-uniform {\BS} problem and its {SDP} relaxation, and how constructing an integrality gap for non-uniform {\BS}
 implies an integrality gap for non-uniform {\SC} and, thus,
disproves the \GLC.

\begin{theorem}[Integrality Gap Instance for Balanced Edge-Separator]\label{thm:main2} Non-uniform {\BS}
has  an integrality gap of
at least $(\log\log n)^{\nfrac{1}{6}-\delta},$ where $\delta > 0$ is an
arbitrarily small constant.  The integrality gap holds for a standard {SDP} relaxation with the
triangle inequality constraints.
\end{theorem}

\begin{theorem}[{\GLC} is False]\label{cor:gl-disprove}  For an arbitrarily small constant $\delta > 0$,  for all sufficiently large
$n$, there is an $n$-point $\ell_2^2$ metric which cannot be
embedded into $\ell_1$ with distortion less than $(\log\log
n)^{\nfrac{1}{6}-\delta}.$
\end{theorem}

\noindent
A surprising aspect of our integrality gap construction is that it proceeds via the Unique Games Conjecture ({UGC}) of Khot \cite{KhotUCSP}  (see Section \ref{sec:UGC} for
the statement of the conjecture). 
 We first prove that the {UGC} implies a
super-constant hardness result for  non-uniform {\BS}. 

\begin{theorem}[UG-Hardness for Balanced Edge-Separator]\label{thm:bs-main} Assuming the Unique Games Conjecture, non-uniform {\BS}  is {NP}-hard to
approximate within any constant factor.
\end{theorem}

\noindent
This particular result was also proved independently
by Chawla {et al.}
\cite{CKKRS}.
Note that this result leads to the following implication:
If the {UGC} is true and P $\not=$ NP, then the {\GLC} must be false!  This
is a rather peculiar situation, because the {UGC} is still
unproven,  and may very well be false.
Nevertheless, we are able to disprove the {\GLC} {\em
unconditionally}.
Indeed, the {UGC} plays a crucial role in our
disproval. Let us outline the  high-level approach we take. First, we
build an integrality gap instance for a natural {SDP}
relaxation of {\UG} (see Figure \ref{sdp:ugc}). We then  {\it translate} this integrality gap instance
into an integrality gap instance of non-uniform {\BS}. This translation {\it mimics} the {PCP} reduction from the
{UGC} to this problem.

\medskip
The integrality gap instance for the
{\UG} {SDP} relaxation (see Figure \ref{sdp:ugc}) is stated below and is one of our main contributions. Here,
we choose to provide an informal description  of this construction
(the reader should be able to understand this construction without even
looking at the {SDP} relaxation).

\begin{theorem}[Integrality Gap Instance for Unique Games- {\em Informal Statement}] \label{thm:ugc-igap} Let $N$ be an integer and
$\eta > 0$ be a parameter (think of $N$ as large and $\eta$ as
tiny). There is a graph $G(V,E)$ of size ~$\nfrac{2^N}{N}$ with the
following properties: Every vertex $u \in V$ is assigned a set of
unit vectors $B(u) \defeq \{\u_1, \ldots, \u_N\}$ that form an
orthonormal basis for the space $\Real^N$. Further,
\begin{enumerate}
\item For every edge $e\{u,v\} \in E$, the sets of vectors $B(u)$
and
      $B(v)$ are almost the same up to some small perturbation. To
      be precise, there is a permutation $\pi_e : [N] \mapsto
      [N],$ such that $\forall \ 1 \leq i \leq N$, $ \  \langle \u_{\pi_e(i)} , \v_{i} \rangle \geq
      1- \eta$. ~In other words, for every edge $(u,v) \in E$, the basis $B(u)$
      moves {\em smoothly/continuously} to the basis $B(v)$.
\item For  any {\em labeling}  $\lambda : V \mapsto [N]$,
       i.e., assignment of an integer $\lambda(u)\in [N]$ to every $u \in V$,
      for at least $1-\nfrac{1}{N^\eta}$ fraction of the edges $e\{u,v\} \in
      E$, we have
      $\lambda(u) \not= \pi_e(\lambda(v))$. In other words, no
      matter how we choose to assign a vector $\u_{\lambda(u)} \in B(u)$ for
      every vertex $u \in V$,
      the movement from $\u_{\lambda(u)}$ to $\v_{\lambda(v)}$ is
      {\em discontinuous} for almost all edges $e\{u,v\} \in E$.
\item All vectors in $\cup_{u \in V} B(u)$ have coordinates in the set
      $\{\nfrac{1}{\sqrt{N}}, \nfrac{-1}{\sqrt{N}}\}$ and, hence, any
      three of them satisfy the triangle inequality constraint.
\end{enumerate}
\end{theorem}

\noindent
This {\UG} integrality gap instance construction is rather non-intuitive (at least to the authors when this paper
 was first written): One can walk on the
graph $G$ by changing the basis $B(u)$ continuously, but as soon
as one picks a {\it representative vector} for each basis, the
motion becomes discontinuous almost everywhere. Of course, one can
pick these representatives in a continuous fashion for any small
enough local sub-graph of $G$, but there is no way to pick
representatives in a global fashion.

\medskip
Before we present a high-level overview of our proofs and discuss the difficulties involved, we 
give a brief overview of related and subsequent works
since the publication of our paper in 2005.

\subsection{Subsequent Works}

For non-uniform {\BS} and, hence,  non-uniform {\SC}, our lower bound was improved to $\Omega(\log \log n)$ by
Krauthgamer and Rabani \cite{KrauthgamerR09} and then to $(\log n)^{\Omega(1)}$ in a sequence of papers by Lee and Naor \cite{LeeN06} and
Cheeger, Kleiner and Naor \cite{CheegerK10, CheegerKN09}.
For the uniform case, Devanur {et al.} \cite{DevanurKSV06} obtained the first super-constant
lower bound of  $\Omega (\log  \log n)$, thus, disproving the ARV conjecture as well. This latter bound has been  recently improved to $2^{\Omega(\sqrt{\log \log n})}$
by Kane and Meka \cite{KaneM13}, building on the {\it short code} construction of Barak {et al.} \cite{BarakGHMRS12}. At a high level,
the constructions in \cite{KrauthgamerR09, DevanurKSV06} are in the same spirit as ours\footnote{Both \cite{KrauthgamerR09} and \cite{DevanurKSV06}  use a result of  Kahn, Kalai and Linial \cite{KahnKL88} instead of Bourgain (Theorem \ref{thm:bourgain-junta}) as in our paper.}  whereas the constructions in \cite{LeeN06, CheegerK10, CheegerKN09} are
entirely different, based on the geometry of Heisenberg group.

\medskip
An unsatisfactory aspect of our construction (and the subsequent ones in \cite{KrauthgamerR09, DevanurKSV06}) is that the feasibility of the triangle
inequality constraints is proved in a brute-force manner with little intuition. A more intuitive proof along with more
general results is obtained by Raghavendra and Steurer \cite{RaghavendraS09} and Khot and Saket \cite{KhotS09}. As a non-embeddability result,
these papers present an $\ell_2^2$ metric that requires super-constant distortion to embed into $\ell_1$, but in addition,
every sub-metric of it on a super-constant number of points is isometrically embeddable into $\ell_1$. The result of Kane and Meka
 also shares this stronger property. We remark that the Kane and Meka
result can be viewed as a derandomization of results in
our paper and those in \cite{KrauthgamerR09, DevanurKSV06, RaghavendraS09, KhotS09}.

\medskip
In hindsight, our paper may be best
viewed as  a scheme that translates a {UGC}-based hardness result into
an integrality gap for a {SDP} relaxation with  triangle inequality constraints.
In the conference version of our paper \cite{KhotV05}, we applied this scheme
to the {\MC} and {\MUC} problems as well. In particular, for
{\MC}, we showed that the integrality gap for the Goemans and Williamson's
{SDP} relaxation \cite{GW} remains unchanged even after adding triangle inequality
constraints. Subsequent works of Raghavendra and Steurer \cite{RaghavendraS09} and Khot and Saket \cite{KhotS09}
cited above extend this paradigm in two directions: Firstly,
their {SDP} solution satisfies additional constraints given
by a super-constant number of rounds of the so-called Sherali-Adams {LP} hierarchy and secondly,
they demonstrate that the paradigm holds for every constraint satisfaction problem ({CSP}).
Since these two works already present more general results and in a more intuitive manner,  we
omit our results for {\MC} and {\MUC} from this paper and keep the overall presentation
cleaner by restricting only to {\SC}.

\medskip
Further, a result of Raghavendra \cite{Raghavendra08} shows that the integrality gap for a certain {\it canonical} {SDP}
relaxation can be {\em translated} into a {UGC}-based hardness result with the same gap (this is
a {\em translation} in the opposite direction as ours). Combined with the results in \cite{RaghavendraS09, KhotS09},
one concludes that the integrality gap for the basic {SDP} relaxation remains unchanged even after
adding a super-constant number of rounds of the Sherali-Adams {LP} relaxation.
Finally, our techniques have inspired
integrality gap for problems that are strictly speaking not {CSP}s,
e.g., integrality gap for the {\QP} problem in \cite{AroraBKSH05, KhotO09}
and some new non-embeddability results, e.g., for the edit distance \cite{KhotNaor}.

\subsection*{Rest of the Introduction}
In Section \ref{sec:intro-overview}, we give a high level overview of
our $\ell_2^2$ vs. $\ell_1$ lower bound. The construction is arguably
unusual and so is the construction of Lee and Naor \cite{LeeN06}
which is based on the geometry of Heisenberg group. 
The latter construction
also needs rather involved mathematical machinery to prove its correctness, see \cite{CheegerKN09}. 
In light of this,
it seems worthwhile to point out the difficulties faced by the
researchers towards proving the lower bound. 
Our discussion in Section \ref{sec:difficulty}
is informal, without precise statements or claims.

\subsection{Difficulty in Proving $\ell_2^2$ vs. $\ell_1$ Lower
Bound}\label{sec:difficulty}

\medskip \noindent {\bf Difficulty in constructing $\ell_2^2$
metrics:} To the best of our knowledge, no {\em natural} or {\em obvious} families of
$\ell_2^2$ metrics are known other than the Hamming metric on
$\{-1,1\}^k$. The Hamming metric is an $\ell_1$ metric and, hence, not
useful for the purposes of obtaining  $\ell_1$ lower bounds. Certain $\ell_2^2$
metrics can be constructed via Fourier analysis and one can also
construct some by solving {SDP}s explicitly. The former approach has
a drawback that metrics obtained via Fourier methods typically
embed into $\ell_1$ isometrically. The latter approach has limited
scope, since one can only  hope to solve {SDP}s of  moderate size.
Feige and Schechtman \cite{FS} show that selecting an appropriate
number of points  from the unit sphere gives an $\ell_2^2$ metric.
However, in this case, most pairs of points have distance
$\Omega(1)$ and, hence, the metric is likely to be
$\ell_1$-embeddable with low distortion.

\medskip \noindent {\bf Difficulty in proving $\ell_1$ lower bounds:}
The techniques to prove an $\ell_1$-embedding lower bound are limited.
To the best of our knowledge, prior to this paper, the only
interesting (super-constant) lower bound was due to  \cite{AR,LLR},
where it is shown that the shortest
path metric on a constant degree expander requires $\Omega(\log
n)$ distortion to embed into $\ell_1$.\footnote{We develop a Fourier analytic technique to
prove an $\ell_1$-embedding lower bound that has been subsequently used in
\cite{KrauthgamerR09, DevanurKSV06, KhotNaor}. The approach of \cite{CheegerK10,CheegerKN09} gives another technique,
by developing an entire new theory of $\ell_1$-differentiability and its
quantitative version. }

\medskip \noindent {\bf General theorems regarding group norms:} A
{\it group norm} is a distance function $d( \cdot,\cdot )$ on a
group $(G,\circ),$ such that $d(x,y)$ depends only on the group
difference $x\circ y^{-1}$. Using Fourier methods, it is possible
to construct group norms that are $\ell_2^2$ metrics. However,
it is known that any group norm on $\mathbb{R}^k,$ or on any group
of characteristic $2,$ is isometrically $\ell_1$-embeddable (see
\cite{DLBook}). Such a result might hold, perhaps allowing a small distortion,
 for every Abelian group
(see \cite{AustinNV10}).  Therefore,
an approach via group norms would probably
not succeed as long as the underlying
group is Abelian. On the other hand, only in the Abelian case,   Fourier methods work well.

\medskip
 The best known lower bounds  for the $\ell_2^2$ versus $\ell_1$
question, prior to this paper,   were due to Vempala ($\nfrac{10}{9}$ for a metric obtained
by a computer search)  and Goemans ($1.024$ for a metric based on the
Leech Lattice), see \cite{SchechtmanSurvey}.  Thus, it appeared that an entirely new approach
was needed to resolve the {\GLC}. In this paper, we present an
approach based on tools from complexity theory, namely, the {UGC}, {PCP}s, and Fourier analysis of Boolean functions.
Interestingly, Fourier analysis is used both to construct the
$\ell_2^2$ metric, as well as, to prove the $\ell_1$ lower bound.

\subsection{Overview of Our $\ell_2^2$ vs. $\ell_1$ Lower Bound} \label{sec:intro-overview}

In this section, we present a high level idea of our $\ell_2^2$
versus $\ell_1$ lower bound, i.e., Theorem \ref{cor:gl-disprove}.
Given the construction of Theorem \ref{thm:ugc-igap}, it is fairly
straight-forward to describe the candidate $\ell_2^2$ metric:
Let $G(V,E)$ be the graph, and $B(u)$ be the orthonormal basis for
$\Real^N$ for every $u \in V$ as in Theorem \ref{thm:ugc-igap}.
  For $u \in V$ and $\x=(x_1,\ldots,x_N) \in \{-1,1\}^N,$ define the vector $\V_{u,x}$ as follows:\footnote{For a vector $x \in \mathbb{R}^N$ and an integer $l,$  the $l$-th tensor of $x,$ $y \defeq x^{\otimes l},$ is a vector in $(\mathbb{R}^N)^l$ defined such that  for $i_1,i_2,\ldots, i_l \in [N],$ $y_{i_1,i_2,\ldots,i_l} \defeq x_{i_1} x_{i_2} \cdots x_{i_l}.$ It follows that for $x,z \in \mathbb{R}^N,$ $\left\langle x^{\otimes l},z^{\otimes l} \right\rangle = \sum_{i_1,i_2,\ldots, i_l \in [N]} (x_{i_1}x_{i_2}\cdots x_{i_l})(z_{i_1}z_{i_2}\cdots z_{i_l})  = \left(\sum_{i \in [N]} x_iz_i  \right)^l =  \langle x,z \rangle ^l.$}
 \begin{equation} \label{eqn:v-usx}
  \V_{u,x}  \defeq  \frac{1}{\sqrt{N}} \sum_{i=1}^N x_i
 \u_i^{\otimes 8}.
   \end{equation}
Note that since $B(u) = \{\u_1,  \ldots, \u_N\}$ is an
orthonormal basis for $\mathbb{R}^N,$  every $\V_{u,x}$ is a unit vector. Fix $t$
to be a large odd integer, for instance $2^{240}+1$, and consider the set
of unit vectors
 $$ {\cal S} \defeq \left\{ \V_{u,x}^{\otimes t} \ | \ u \in V, \ x
 \in \{-1,1\}^N \right\}. $$
Using, essentially,  the fact that the
  vectors in $\cup_{u \in
V} B(u)$ are a {\em good} solution to the {SDP} relaxation of {\UG},
we are able to show that every triple of vectors in ${\cal S}$ satisfy the
triangle inequality constraint and, hence, $\mathcal{S}$ defines an $\ell_2^2$ metric. One can also directly  show  that this $\ell_2^2$
metric does not embed into $\ell_1$ with distortion less than $(\log
N)^{\nfrac{1}{6}-\delta}$.

\medskip
However, we choose to present our construction in a different and
an  indirect way.  The (lengthy) presentation goes through the
UGC and the {PCP} reduction from {\UG}
integrality gap instance to {\BS}. Hopefully, our presentation
 brings out the intuition as to why and how we came up with the
above set of vectors, which happened to define an $\ell_2^2$
metric. At the end, the reader should recognize that the idea of
taking all $+/-$ linear combinations of vectors in $B(u)$ (as in
Equation \eqref{eqn:v-usx}) is directly inspired by the {PCP}
reduction. Also, the proof of the $\ell_1$ lower bound is
hidden inside the {\it soundness analysis} of the {PCP}.

\medskip
\noindent
The overall construction can be divided into three steps:
\begin{enumerate}
\item A {PCP} reduction from {\UG} to {\BS}.       \item Constructing an integrality gap instance for a natural {SDP}
relaxation of {\UG}.
\item Combining the above two to construct an
integrality gap instance of {\BS}. This also
gives an $\ell_2^2$ metric that needs $(\log \log n)^{\nfrac{1}{6}-\delta}$  distortion
to embed into $\ell_1$.
\end{enumerate}

\noindent
We present an overview of each of these steps in three separate
sections. Before we do that, let us summarize the precise notion
of an integrality gap instance of {\BS}. To keep things simple in
this exposition, we  pretend as if our construction works for
the uniform version of {\BS} as well. (Actually it does not; we
have to work with the non-uniform version which complicates things a
little.)

\subsubsection*{{SDP} Relaxation of Balanced Edge-Separator}

Given a graph $G'(V',E')$, {\BS} asks
for a $(\nfrac{1}{2}, \nfrac{1}{2})$-partition of $V'$ that cuts as
few edges as possible (however, the algorithm is allowed to output
a roughly balanced partition, say $(\nfrac{1}{4},
\nfrac{3}{4})$-partition). We  denote an edge $e$ between vertices $i,j$ by $e\{i,j\}.$ The {SDP} relaxation of
{\BS} appears in Figure \ref{fig:usdp-bs}.

\begin{figure}[hbt]
\begin{equation} \label{bl-usdp-1}
  \mbox{Minimize} \ \ \ \frac{1}{|E'|}  \sum_{e'\{i,j\} \in E'}  \frac{1}{4} \|\v_i -
  \v_j  \|^2
\end{equation}
Subject to
\begin{eqnarray}
\forall  \  i \in V' &  \| \v_i \|^2 = 1 \label{bl-usdp-2} \\
\forall \ i,j,l\in V'  &  \| \v_i - \v_j \|^2 + \|\v_j - \v_l\|^2 \geq \| \v_i - \v_l\|^2  \label{bl-usdp-3}\\
 & \sum_{i < j}  \| \v_i - \v_j\|^2 \geq |V'|^2
 \label{bl-usdp-4}
\end{eqnarray}
\caption{ {SDP} relaxation of the uniform version of {\BS}}
 \label{fig:usdp-bs}
\end{figure}

\noindent
Note that a $\{+1,-1\}$-valued solution represents a true
partition and, hence, this is an {SDP} relaxation. Constraint
\eqref{bl-usdp-3} is the triangle inequality constraint and
Constraint \eqref{bl-usdp-4} stipulates that the partition be
balanced.\footnote{Notice that if a set of vectors  $\{v_i: i \in V'\}$ is such that for every vector in the set, its antipode is also in  the set, then constraint \eqref{bl-usdp-4} is automatically satisfied. Our construction   obeys this property.} 
The notion of integrality gap is summarized in the following
definition:

\begin{defn} \label{defn:bs-igap}
An integrality gap instance of {\BS} is a  graph
$G'(V',E')$ and an assignment of unit vectors $i \mapsto \v_i$ to
its vertices such that:
\begin{itemize}
\item Every  balanced partition (say $(\nfrac{1}{4},
\nfrac{3}{4})$-partition, this choice is arbitrary)  of $V'$ cuts at least $\alpha$ fraction of edges. \item  The set of vectors
$\{\v_i | \ i \in V'\}$ satisfy
\eqref{bl-usdp-2}-\eqref{bl-usdp-4}, and the {SDP} objective value in
Equation \eqref{bl-usdp-1} is at most $\gamma$.
\end{itemize}
The integrality gap is defined to be $\nfrac{\alpha}{\gamma}$ (thus, we
desire that $\gamma \ll \alpha$).
\end{defn}

\noindent
 The next three sections describe the three steps involved
 in constructing an integrality gap instance of {\BS}. Once that is done, it follows from a folk-lore
 result that the resulting $\ell_2^2$ metric (defined by
 vectors $\{\v_i | \ i \in V'\}$) requires distortion at least
 $\Omega(\nfrac{\alpha}{\gamma})$ to embed into $\ell_1$. This would prove
 Theorem \ref{cor:gl-disprove} with an appropriate choice of
 parameters.

\subsubsection*{The {PCP} Reduction from Unique Games  to Balanced Edge-Separator}
An instance $\ \calU=(G(V,E), [N],
\{\pi_e\}_{e \in E})$  of {\UG}  consists of a graph $G(V,E)$ and
permutations $\pi_e : [N] \mapsto [N]$ for every edge $e\{u,v\} \in
E$. The goal is to find a {\em labeling} $\lambda : V \mapsto [N]$
 that {\em satisfies} as many edges as possible. An edge
$e\{u,v\}$ is satisfied if $\lambda(u) = \pi_e(\lambda(v))$. Let
${\rm opt}(\calU)$ denote the maximum fraction of edges satisfied by any
labeling.

\medskip
\noindent {\it {UGC} (Informal Statement): It is {NP}-hard to
decide whether an instance  $\calU$   of {\UG} has ${\rm opt}(\calU)
\geq 1- \eta$ ({YES} instance) or ${\rm opt}(\calU)\leq \zeta$ ({NO}
instance), where $\eta, \zeta>0$ can be made arbitrarily small by
choosing $N$ to be a sufficiently large constant. }

\medskip
\noindent
It is possible to construct an instance of {\BS} $G'_\epsilon(V',E')$ from an instance
of {\UG}. We describe only the high
level idea here. The construction is parameterized by
$\epsilon>0$. The graph $G'_\epsilon$ has a block of $2^N$ vertices
for every $u \in V$. This block contains one vertex for every point
in the Boolean hypercube $\{-1,1\}^N$. Denote the set of these
vertices by $V'[u].$ More precisely,
  $$ V'[u]  \defeq \left\{ (u, x) \ | \  x \in \{-1,1\}^N \right\}.$$
We let $V'  \defeq \cup_{u \in V} V'[u]$. For every edge $e\{u,v\} \in
E$, the graph $G'_\epsilon$ has edges between the blocks $V'[u]$
and $V'[v]$. These edges are supposed to capture the constraint
that the labels of $u$ and $v$ are consistent, i.e., $\lambda(u) =
 \pi_e( \lambda(v) )$. Roughly speaking, a vertex $(u,x) \in
V'[u]$ is connected to a vertex $(v,y) \in V'[v]$ if and only if,  after
identifying the coordinates in $[N]$ via the permutation $\pi_e$, the
 Hamming
distance between the bit-strings $x$ and $y$ is about
$\epsilon N$.
This reduction has the following two properties:
\begin{theorem} \label{thm:bs-pcp-informal} ({PCP} reduction: Informal statement)
\begin{enumerate}
\item (Completeness/{YES} case):  If ${\rm opt}(\calU) \geq 1-\eta$, then
the
 graph $G'_\epsilon$ has a $(\nfrac{1}{2}, \nfrac{1}{2})$-partition
 that cuts at most $\eta+\epsilon$ fraction of its edges.
\item (Soundness/{NO} Case): If ${\rm opt}(\calU) \leq
 2^{-O(\nfrac{1}{\epsilon^2})}$, then every $(\nfrac{1}{4}, \nfrac{3}{4})$-partition
 of $G'_\epsilon$ cuts at least $\sqrt{\epsilon}$
 fraction of its edges.
\end{enumerate}
\end{theorem}

\begin{remark} We were imprecise on two counts: (1)
The soundness property holds only for those partitions that
partition a constant fraction of the blocks $V'[u]$ in a roughly
balanced way. We call such partitions {\em piecewise balanced}. This
is where the issue of uniform versus non-uniform version of
{\BS} arises. (2) For the soundness property, we
can only claim that every piecewise balanced partition cuts at
least $\epsilon^t$ fraction of edges, where any $t > \nfrac{1}{2}$
can be chosen in advance. Instead, we write $\sqrt{\epsilon}$ ~for the
simplicity of notation.
\end{remark}

\subsubsection*{Integrality Gap
 Instance for the Unique Games {SDP} Relaxation}

This has  already been described in Theorem \ref{thm:ugc-igap}.
The graph $G(V,E)$ therein along with the orthonormal basis
$B(u),$ for every $u \in V,$ can be used to construct an instance
$\calU=(G(V,E), [N], \{\pi_e\}_{e \in E})$ of {\UG}. For every edge
$e\{u,v\} \in E$, we have an (unambiguously defined)  permutation
$\pi_e : [N] \mapsto [N],$ where $\ \langle \u_{\pi_e(i)}, \v_i
\rangle \geq 1-\eta,$ for all $1 \leq i \leq N$.

\smallskip
Theorem \ref{thm:ugc-igap} implies that ${\rm opt}(\calU) \leq
\nfrac{1}{N^\eta}$.  On the other hand, the fact that for every
edge $e\{u,v\}$, the bases $B(u)$ and $B(v)$ are very close to each
other means that the {SDP} objective value for $\calU$ is at least $1-\eta$ (formally, the {SDP} objective value is defined to be
$\E_{e\{u,v\} \in E} \left[ \frac{1}{N} \sum_{i=1}^N \langle \u_{\pi_e(i)}, \v_i \rangle\right]$).

\smallskip
Thus, we have a concrete instance of {\UG} with optimum at
most $\nfrac{1}{N^\eta}=o(1)$, and which has an {SDP} solution with
objective value at least $ 1-\eta$. This is what an integrality gap
example means: The {SDP} solution {\em cheats} in an {\it unfair} way.

\subsubsection*{Integrality Gap Instance for the Balanced Edge-Separator {SDP} Relaxation} Now we combine the two modules described above.  We take
the instance $\calU=(G(V,E), [N], \{\pi_e\}_{e \in E})$ as above
and run the {PCP} reduction on it. This gives us an instance
$G'(V',E')$ of {\BS}. We show that this
is an integrality gap instance in the sense of Definition
\ref{defn:bs-igap}.

\smallskip
Since $\calU$ is a {NO} instance of {\UG}, i.e., ${\rm opt}(\calU) =
o(1)$, Theorem \ref{thm:bs-pcp-informal} implies that every
(piecewise) balanced partition of $G'$ must cut at least
$\sqrt{\epsilon}$ fraction of the edges. We need to have $\nfrac{1}{N^\eta}
\leq 2^{-O(\nfrac{1}{\epsilon^2})}$ for this to hold.

\smallskip
On the other hand, we can construct an {SDP} solution for  the {\BS}  instance which has an objective value of at most
$O(\eta+\epsilon)$. Note that a typical vertex of $G'$ is
$(u,x),$ where  $u \in V$ and $x \in \{-1,1\}^N$. To this
vertex, we attach the unit vector $\V_{u,x}^{\otimes t}$ (for
$ t=2^{240}+1$), where
 $$\V_{u,x}   \defeq  \frac{1}{\sqrt{N}} \sum_{i=1}^N x_i
 \u_i^{\otimes 8}.$$

\smallskip
\noindent
It can be shown that the set of  vectors $\left\{\V_{u,x}^{\otimes
t} \ | \ u \in V, \ x \in \{-1,1\}^N \right\}$ satisfy the triangle
inequality constraint and, hence, defines an $\ell_2^2$ metric.
Vectors $\V_{u,x}^{\otimes t}$ and $\V_{u,-x}^{\otimes
t}$ are antipodes of each other and, hence, the {SDP} Constraint
\eqref{bl-usdp-4} is also satisfied. Finally, we show that the
{SDP} objective value (Expression \eqref{bl-usdp-1}) is
$O(\eta+\epsilon)$. It suffices to show that for every edge
$((u,x), (v, y))$ in $G'(V',E')$, we have
 $$ \left \langle \V_{u,x}^{\otimes t} , \V_{v,y}^{\otimes t} \right \rangle \
 \geq  \ 1-O(t(\eta+\epsilon)). $$
This holds because whenever $((u,x), (v, y))$ is an edge of
$G'$, we have (after identifying the  indices via the permutation
$\pi_e : [N] \mapsto [N]$):
\begin{enumerate}
\item $\langle \u_{\pi_e(i)}, \v_i \rangle \geq 1-\eta$  for all $1 \leq i \leq N$ and  
\item the Hamming distance between $x$
and $y$ is about $\epsilon N$.
\end{enumerate}

\subsubsection*{Quantitative Parameters}
It follows from above discussion (see also Definition
\ref{defn:bs-igap}) that the integrality gap for {\BS} is $\Omega(\nfrac{1}{\sqrt{\epsilon}})$ provided that  $\eta \approx \epsilon$, and $N^\eta > 2^{O(\nfrac{1}{\epsilon^2})}$. We can
choose $\eta \approx \epsilon \approx (\log N)^{\nfrac{-1}{3}}$. Since the size of
the graph $G'$ is at most $n=2^{2N}$, we see that the integrality
gap is $\approx (\log\log n)^{\nfrac{1}{6}}$ as desired.

\subsubsection*{Proving the Triangle Inequality}

As mentioned above, one can show that the set of vectors
$\{\V_{u,x}^{\otimes t} \ | \ u \in V, \ x \in \{-1,1\}^N
\}$ satisfy the triangle inequality constraints. This is the
most technical part of the paper, but we would like to stress that this is
where the {\em magic} happens. In our construction, all vectors in
$\cup_{u \in V} B(u)$ happen to be points of the hypercube
$\{-1,1\}^N$ (up to a normalizing factor of $\nfrac{1}{\sqrt{N}}$), and
therefore, they define an $\ell_1$ metric. The
operation that takes their $+/-$ combinations combined with  tensoring leads to a metric that is $\ell_2^2$ and
non-$\ell_1$-embeddable.

\smallskip
Our proof of the triangle inequality constraints is essentially brute-force. As we mentioned before, more
recent works \cite{RaghavendraS09, KhotS09} obtain a more intuitive proof.

\subsection{Organization of the Main  Body of the Paper}\label{sec:organization}
In Section \ref{sec:prelim-metric-defns}
we recall important  definitions and
results about metric spaces. Section \ref{sec:prelim-cuts} defines
the cut optimization problems we are concerned about: {\SC} and {\BS}. We also give
their {SDP} relaxations for which we  construct
integrality gap instances. Section \ref{sec:prelim-fourier}
presents useful tools from Fourier analysis.

\medskip
\noindent In Section \ref{sec:GLC}, we present our overall
strategy for disproving the {\GLC}. We give a disproval of the {\GLC} assuming an appropriate integrality gap  for {\BS}.

\medskip
\noindent
In Section \ref{sec:UGC} we present the {UGC} and our
integrality gap instance for an  {SDP} relaxation of {\UG}.

\medskip
\noindent
In Section \ref{sec:inner-bs} we present our PCP reduction from {\UG} to {\BS}. The soundness proof this reduction is standard and appears in   Appendix \ref{sec:app-pcp}.

\medskip
\noindent We build on the {\UG} integrality gap instance in
Section \ref{sec:UGC} and the PCP reduction in Section \ref{sec:inner-bs} to obtain the integrality gap
instance for {\BS}. This is presented in Section \ref{sec:INTGAP}. This section has two
parts: In the first part (Section \ref{sec:INTGAP-graphs}) we
present the graph and in the second part (Section
\ref{sec:INTGAP-sdp}) we present the corresponding {SDP}
solution and prove its properties.

\medskip
\noindent Appendix \ref{sec:sdp-triangle} is where we establish the main technical lemma needed to show that the {SDP} solutions we construct satisfy the triangle inequality constraint.

\section{Preliminaries}\label{sec:prelim}

\subsection{The {\GLC} }\label{sec:prelim-metric-defns}
We start with basics of metric embeddings. We are concerned with finite metric spaces which we denote by a pair $(X,d),$ where $X$ is the space and $d$ is the metric on its points. We say that a space $(X_1,d_1)$ embeds with distortion at most $\Gamma$ into another space
$(X_2,d_2)$ if there exists a map $\phi: X_1 \mapsto X_2$ such
that for all $x,y \in X_1$
$$ d_1(x,y) \leq d_2 (\phi(x),\phi(y)) \leq \Gamma \cdot d_1(x,y).$$
If $\ \Gamma=1,$ then $(X_1,d_1)$ is said to {\em isometrically}
embed in $(X_2,d_2).$

\smallskip
An important class of metric spaces are those that arise by taking a finite subset $X$ of $\mathbb{R}^m$ for some $m\geq 1$ and endowing it with the $\ell_p$ norm as
 follows: For $\x=(x_1,\ldots,x_m),\y=(y_1,\ldots,y_m) \in X,$ $$\ell_p(\x,\y) \defeq\left(
\sum_{i=1}^m |x_i-y_i|^p \right)^{\nfrac{1}{p}}. $$  When we call a metric $\ell_1$ or $\ell_2,$ an implicit underlying space is assumed.

\smallskip
 A metric space $(X,d)$ is said
to be of negative type  if $(X,\sqrt{d})$ embeds isometrically into
$\ell_2.$ Formally, there is an integer $m$ and a vector $\v_x \in
\mathbb{R}^m$ for every $x \in X,$ such that $d(x,y)=\|\v_x-\v_y
\|^2$ and the vectors satisfy the  {\em triangle inequality}, i.e., for all $x,y,z \in X,$   
$$ \|v_x - v_y\|^2 + \| v_y-v_z\|^2 \geq \| v_x-v_z\|^2.$$ 
The class of all negative type metrics is denoted by $\ell_2^2.$
The following fact is easy to prove.
\begin{fct}\label{fct:l1neg}\cite{DLBook} For every $\ell_1$ metric space $(X,\ell_1)$ there is a negative type metric space $(Y,d)$ in which  it embeds isometrically.
\end{fct}

\noindent
While the converse is not true, the {\GLC} asserts that the converse holds up to a universal constant.

\begin{conjecture}[{\GLC}, \cite{GoemansSurvey,LinialSurvey}] \label{hypo:gl}
For every negative type metric space $(Y,d)$ there is a metric space $(X,\ell_1)$ in which  it embeds with at most a constant distortion. This constant is universal, i.e., independent of the metric space $(Y,d)$.
\end{conjecture}

\subsection{Balanced Edge-Separator, Sparsest Cut and their {SDP} Relaxations}\label{sec:prelim-cuts}
In this section, we define the {\BS} and the {\SC} problems and their {SDP} relaxations.
All graphs are complete
undirected graphs
with non-negative {\em weights} or {\em demands} associated to its
edges. For a graph $G(V,E)$ and $S \subseteq V$, let $E(S,
\overline{S})$ denote the set of edges with one endpoint in $S$
and other in $\overline{S}$. A cut $(S, \overline{S})$ is called
non-trivial if $S \not=\emptyset$ and $ \overline{S}\not= \emptyset$.

\begin{remark} The versions of {\SC} and {\BS}
that we define below are  non-uniform  versions with demands. The
uniform version has all demands equal to $1,$ i.e., unit
demand for every pair of vertices.
\end{remark}

\begin{defn}[{\SC}]
For a graph $G(V,E)$ with a  weight ${\rm wt}(e)$ and a demand ${\rm dem}(e)$ associated to each edge $e \in E,$
the goal is to optimize
$$ \min_{\emptyset \neq S \subsetneq V}  \frac{\sum_{e \in E(S, \overline{S})} {\rm wt}(e) }
                    {\sum_{e \in E(S, \overline{S})}  {\rm dem}(e) } .$$
For a cut  $(S, \overline{S})$, the ratio above is referred to as its sparsity.
\end{defn}

\medskip\noindent The SDP relaxation for {\SC} appears in Figure \ref{fig:sdpsc}. We note that this is indeed a relaxation: Any cut $(S, \overline{S})$ corresponds to a feasible SDP solution by setting
the vector $\v_x$ to be ${v_0}$ or $- {v_0}$ depending on whether $x \in S$ or $x \in \overline{S}$
and ${v_0}$ is some fixed vector. The length of ${v_0}$ is chosen so as to satisfy the last SDP constraint.
The SDP objective is then the same as the sparsity of the cut.

\begin{figure}[hbt]
\begin{equation*}
  \mbox{Minimize} \ \ \ \frac{1}{4}  \sum_{e\{x,y\}} {\rm wt}(e) \|\v_x - \v_y
  \|^2
\end{equation*}
Subject to
\begin{eqnarray*}
\forall \ x,y,z\in V  &  \| \v_x - \v_y \|^2 + \|\v_y - \v_z\|^2 \geq \| \v_x - \v_z \|^2 \\
&  \frac{1}{4} \sum_{e\{x,y\}} {\rm dem}(e) \| \v_x - \v_y\|^2=1
\end{eqnarray*}
\caption{SDP relaxation of  {\SC}}\label{fig:sdpsc}
\end{figure}

\noindent
The integrality gap of this SDP relaxation is defined to be the largest ratio, as a function of the number of vertices $n$
and over all possible instances, between the integral optimum and the SDP optimum.
It is known (folklore) that the integrality gap $f(n)$ of the {\SC} SDP relaxation is precisely the worst case
distortion incurred to embed an $n$-point $\ell_2^2$ metric into $\ell_1$. We  need this observation
(but only in one direction) in what follows. First, we formally introduce {\BS}.

\begin{defn}[{\BS}] For a graph $G(V,E)$ with a
weight ${\rm wt}(e),$ and a demand ${\rm dem}(e)$ associated to
each edge $e \in E,$ let $D  \defeq \sum_{e \in E} {\rm dem}(e)$ be the
total demand. Let a {\em balance} parameter $B$ be given where
$\nfrac{D}{6} \leq B \leq \nfrac{D}{2}$. The goal is to find a non-trivial cut
$(S,\overline{S})$ that minimizes
$ \sum_{e \in E(S,
\overline{S}) } {\rm wt}(e),$ subject to
 $\sum_{e \in E(S, \overline{S}) } {\rm dem}(e)
 \geq B.$
The cuts that satisfy $\sum_{e \in E(S, \overline{S}) }
{\rm dem}(e) \geq B$ are called $B$-{balanced cuts}.
\end{defn}

\noindent
The SDP relaxation for  {\BS}  appears in Figure \ref{fig:sdpbs}. We note that this is indeed a relaxation:
A $B$-balanced cut $(S, \overline{S})$ corresponds to a feasible SDP solution by setting
the vector $\v_x$ to be ${v_0}$ or $- {v_0}$ depending on whether $x \in S$ or $x \in \overline{S}$
and ${v_0 }$ is a fixed {\it unit} vector.
\begin{figure}[hbt]
\begin{equation} \label{bl-sdp-1}
  \mbox{Minimize} \ \ \ \frac{1}{4}  \sum_{e\{x,y\}} {\rm wt}(e) \|\v_x - \v_y
  \|^2
\end{equation}
Subject to
\begin{eqnarray}
\forall x \in V &  \| \v_x \|^2 = 1 \label{bl-sdp-2} \\
\forall \ x,y,z\in V  &  \| \v_x - \v_y \|^2 + \|\v_y - \v_z\|^2 \geq \| \v_x - \v_z \|^2  \label{bl-sdp-3}\\
&  \frac{1}{4} \sum_{e\{x,y\}} {\rm dem}(e) \| \v_x - \v_y\|^2
\geq B
 \label{bl-sdp-4}
\end{eqnarray}
\caption{SDP relaxation of  {\BS} with parameter $B$}\label{fig:sdpbs}
\end{figure}

\medskip
An integrality gap instance for {\BS} is a concrete instance along with a feasible $B$-balanced SDP solution such that
the SDP objective is at most $\gamma$ and the integral optimum  over $\nfrac{B}{3}$-balanced
cuts is at least $\alpha$. The integrality gap is $\nfrac{\alpha}{\gamma}$. Note that the SDP solution is
$B$-balanced (in the sense of the last SDP constraint), but the integral optimum is allowed over
$\nfrac{B}{3}$-balanced cuts, i.e., over a larger class of cuts than the $B$-balanced cuts.

\subsection{Relation Between {\GLC}, Sparsest Cut and Balanced Edge-Separator}\label{sec:sc-bs}

Consider the following three statements:
\begin{enumerate}
\item Every $n$-point $\ell_2^2$ metric embeds into $\ell_1$ with distortion at most $f(n)$.
\item The integrality gap of the {\SC} SDP relaxation is at most $f(n)$.
\item The integrality gap of the {\BS} SDP relaxation is at most $O(f(n))$.
\end{enumerate}

\noindent It is known (folklore) that $(1) \implies (2) \implies (3)$ (and in fact $(1)$ is equivalent to $(2)$). We use the
implication $(1) \implies (3)$ to conclude our $\ell_2^2$ vs. $\ell_1$ lower bound from our integrality gap construction for
{\BS}. We summarize this implication below and present a sketch of its proof for the sake of completeness. The proof implicitly
also proves the implication $(1) \implies (2)$.

\begin{lemma}\label{lem:existence-cut}
Suppose $x \mapsto \v_x$ is a solution for {SDP} of Figure
\ref{fig:sdpbs} with objective value
$$\frac{1}{4}\sum_{e\{x,y\}} {\rm wt}(e) \|\v_x - \v_y
\|^2 \leq  \epsilon. $$ Assume that the negative type metric
defined by the vectors $\{\v_x | \ x \in V\}$ embeds into $\ell_1$ with
distortion $f(n)$ where $n= |V|$.  Then, there exists a $B'$-balanced cut
$(S,\overline{S})$, $B'\geq \nfrac{B}{3}$ such that
$$\sum_{e \in
E(S, \overline{S})}{\rm wt}(e) \leq O(f(n) \cdot\epsilon).$$
\end{lemma}

\begin{proof} The idea is that the {\it good} SDP solution as given implies the existence of a cut with low sparsity.
If this cut already cuts $\Omega(B)$ of the demands, we are done. Otherwise the demands cut are {\it erased} (i.e.,
 set to zero) and another sparse cut is found w.r.t. to the new (remaining) demands. This process is repeated until the sum of the demands
cut in the sequence of cuts obtained so far is at least $\Omega(B)$. At this point, a random XOR
of the cuts obtained so far yields a cut that cuts $\Omega(B)$ of the demands, but does not cut too much of the
edge weight.
Formally, we begin by observing that there is a cut $(S,\overline{S})$ with sparsity at most $f(n)\cdot\nfrac{\epsilon}{B}.$
\begin{eqnarray*}
\min_{\emptyset \neq S \subsetneq V}  \frac{\sum_{e \in E(S, \overline{S})} {\rm wt}(e) }
                    {\sum_{e \in E(S, \overline{S})}  {\rm dem}(e) } &  = & \min_{d\; {\rm is} \;  \ell_1 \; {\rm embeddable} }  \frac{\sum_{e\{x,y\}} {\rm wt}(e)d(x,y) }
                    {\sum_{e\{x,y\}}{\rm dem}(e)d(x,y) } \\
& \leq  & f(n) \cdot  \frac{\sum_{e\{x,y\}} {\rm wt}(e) \| v_x - v_y \|^2 }
                    {\sum_{e\{x,y\}}{\rm dem}(e) \| v_x - v_y \|^2 }  \ \leq  \ f(n)\cdot\nfrac{\epsilon}{B}. 
\end{eqnarray*}

\noindent
The first (in)equality uses the fact that optimizing over cuts is the same as optimizing over the cone of  $\ell_1$ embeddable metrics, see \cite{DLBook}.
The second inequality uses the embedding of the metric $\| v_x - v_y \|^2$ into $\ell_1$ with distortion at most $f(n)$. The third inequality uses the
hypothesis that the SDP objective is at most $\epsilon$ and the SDP solution is $B$-balanced.

\smallskip
 If the cut $(S,\overline{S})$ happens to be
$\nfrac{B}{3}$-balanced, then we are done since the edge weight cut by it is at most the sparsity (which is at most   $f(n)\cdot\nfrac{\epsilon}{B}$)
times the demands cut (which is at most $D \leq 6B$).
Otherwise the demands cut by $(S,\overline{S})$ is at most $\nfrac{B}{3}$. We rename the cut as $(S_1,\overline{S}_1)$, set all the demands
cut to zero, and repeat the process. This leads to a sequence of cuts  $(S_1,\overline{S}_1), \ldots, (S_k,\overline{S}_k)$. The
process stops as soon as either 
\begin{enumerate}
\item[(a)] the cut just obtained cuts at least  $\nfrac{B}{3}$ of the demands
or else 
\item[(b)] the sum of the demands cut over these $k$ cuts
is at least $\nfrac{2B}{3}$ (since a demand is set to zero as soon as it is cut, each original demand is counted at most once).
\end{enumerate}
Note that prior to every step, at most $\nfrac{2B}{3}$ of the (original) demands
has been set to zero, so the SDP solution w.r.t. to the remaining demands still qualifies as being $B - \nfrac{2B}{3} =  \nfrac{B}{3}$ balanced.
Thus, at every step, the cut obtained has sparsity at most  $f(n)\cdot\nfrac{\epsilon}{(\nfrac{B}{3})}.$
We are done in the Case (a) as before and so we consider the Case (b).

\smallskip
To summarize, we have a sequence of cuts  $(S_1,\overline{S}_1), \ldots, (S_k,\overline{S}_k)$ such that  the sum of the demands cut over these $k$ cuts
is at least $\nfrac{2B}{3}$.
Moreover, the sparsity of each of these cuts is at most $O( f(n)\cdot\nfrac{\epsilon}{B})$
and, hence, the total edge weight cut by these cuts is at most $O(f(n) \epsilon)$ (an edge is considered cut if it is cut by at least
one of the $k$ cuts). Now we obtain our desired {\it balanced partition} by taking a random XOR
of these cuts: The $i$-{th} cut is viewed as a $\{0,1\}$-valued function $\phi_i$ on the vertices and the desired cut is
given by the function $\phi_A \defeq \oplus_{i\in A} \phi_i$ where $A \subseteq [k]$ is a uniformly random subset. We show that
for some choice of the set $A$, we get a cut $\phi_A$ that
cuts at least  $\nfrac{B}{3}$ of the demands and  at most $O(f(n) \epsilon)$ of the edge weight. Clearly, the total
edge weight cut is $O(f(n) \epsilon)$ irrespective of the set $A$. On the other hand, each  demand in the sum total of
at least $\nfrac{2B}{3}$ gets cut with probability $\nfrac{1}{2}$ (this is the property of the random XOR).
 Thus, the expected demands cut by $\phi_A$ is at least
$\nfrac{B}{3}$ and this expectation is achieved for some choice of $A$.

\end{proof}

\begin{remark} \label{rem:sc-bs}
The proof above shows that if the integrality gap for {\SC} is upper bounded by $f(n)$ then the gap for {\BS} is bounded by
$O(f(n))$. The same proof implicitly also shows that if there is an $f(n)$ approximation algorithm for {\SC}, then the algorithm
can be used iteratively a polynomial number of times to achieve $O(f(n))$ (pseudo-)approximation for {\BS}, see also \cite[Chapter 7]{Vishnoi12}. Given an instance of
{\BS} that has a $B$-balanced cut that cuts an edge weight $\alpha$ and $B \geq \nfrac{D}{6}$ where $D$ is the total demand, the
algorithm finds a $\nfrac{B}{3}$-balanced cut that cuts an edge weight $O(f(n)\alpha)$. In the contrapositive, a $g(n)$ {\it hardness} of
approximation result for {\BS} implies an $\Omega(g(n))$ hardness result for {\SC}.
\end{remark}

\subsection{Our Integrality Gap Instance for Balanced Edge-Separator}\label{sec:GLC}
With the preliminaries for negative type metrics and {SDP}s in place, we now state the main result regarding the construction of the integrality gap for {\BS} which
 suffices to disprove the {\GLC}
using Lemma \ref{lem:existence-cut}.
 The instance  has two parts: (1) The graph and (2) The {SDP} solution. The graph  construction is described in Section \ref{sec:INTGAP-graphs}, while the {SDP} solution appears in Section \ref{sec:INTGAP-sdp}.
We construct a complete weighted graph $G(V,{\rm wt}),$ with vertex
set $V$ and weight ${\rm wt}(e)$ on edge $e,$ and with $\sum_{e} {\rm
wt}(e)  = 1$. The vertex set is partitioned into sets $V_1, V_2,
\ldots, V_r$, each of size $\nfrac{|V|}{r}$ (think of $r \approx
\sqrt{|V|}$).
A cut $A$ in the graph is viewed as a function $A : V \mapsto \{-1,1\}$.
We are interested in cuts that cut {\em many} sets $V_i$ in a {\em somewhat
balanced} way. The notation $s \in _R S$ would mean that $s$ is a uniformly random element of $S.$
\begin{defn} \label{def:part-bl}
For $ \  0 \leq \theta \leq 1$, a cut $A: V \mapsto \{-1,1\}$ is called $\theta$-piecewise
balanced if
 $$ \E_{i \in _R [r]}    \ \Big| \ \E_{x \in_R V_i}[ A(x) ] \ \Big|
 \leq \theta. $$
\end{defn}
\noindent
We also assign a unit vector to every vertex in the graph. Let $\v_x$ denote the
vector assigned to vertex $x$. Our construction of the graph $G(V, {\rm wt})$
and the vector assignment $x \mapsto \v_x$ can be summarized as follows:

\begin{theorem}[Main Theorem]\label{thm:graph-construction} Fix any
$\nfrac{1}{2} < t < 1$. For every sufficiently small $\epsilon > 0,$  there exists a
graph $G(V, {\rm wt})$, with a partition  ~$V = \cup_{i=1}^r
V_i$, and a vector assignment $x \mapsto \v_x$  for every $x \in
V,$ such that
\begin{enumerate}
\item $|V| \leq 2^{2^{O(\nfrac{1}{\epsilon^3})}}$.
\item Every
$\nfrac{5}{6}$-piecewise balanced cut $A$ must cut $\epsilon^{t}$
      fraction of edges, i.e., for any such cut
      $$ \ \ \sum_{e \in E(A, \overline{A})} {\rm wt}(e)  \geq \ \epsilon^{t}. $$
\item The unit vectors $\{\v_x \ | \ x \in V \}$ define a negative
type metric, i.e.,
      the following triangle inequality is satisfied:
      $$  \forall \ x,y,z  \in V , \ \| \v_x - \v_y \|^2 + \| \v_y - \v_z \|^2 \geq \| \v_x - \v_z \|^2 \
        .$$
\item For each part $V_i$, the vectors $\{ \v_x \ | \ x \in V_i\}$ are
      {\em well-separated}, i.e.,
      $$ \frac{1}{2}  \E_{x,y \in_R V_i} \ \left[  \| \v_x -  \v_y \|^2 \right] = 1 .$$
\item The vector assignment gives a {\em low} {SDP} objective value, i.e.,
      $$  \frac{1}{4}
  \sum_{e\{x,y\}} {\rm wt}(e)  \| \v_x - \v_y \|^2  \leq \epsilon .$$
\end{enumerate}
\end{theorem}

\paragraph{Proof of Theorem \ref{cor:gl-disprove}.} We show how the construction in Theorem \ref{thm:graph-construction} implies
Theorem \ref{cor:gl-disprove}. Suppose that
the negative type metric defined by
vectors $\{\v_x | \ x \in V \}$ 
embeds into $\ell_1$ with distortion $\Gamma$.
 We  show that $\Gamma =
\Omega\left(\nfrac{1}{ \epsilon^{1-t}  }\right)$ using
 Lemma \ref{lem:existence-cut}.

\smallskip
Construct an instance of {\BS} as follows.
The graph $G(V, {\rm wt})$ is as in Theorem
\ref{thm:graph-construction}. The demands ${\rm dem}(e)$ depend on the
partition $V = \cup_{i=1}^r V_i$. We let ${\rm dem}(e) =1 $ if $e$ has
both endpoints in the same part $V_i$ for some $1 \leq i \leq r$
and ${\rm dem}(e)=0$ otherwise. Clearly, the total demand is $D \defeq\sum_e {\rm dem}(e)
 =   r\cdot \binom{\nfrac{|V|}{r}}{2}$.

\smallskip
Now, $x \mapsto \v_x$ is an assignment of unit vectors that
satisfy the triangle inequality constraints. This is  a solution to the {SDP} of Figure \ref{fig:sdpbs}. Property (4) of Theorem
\ref{thm:graph-construction} guarantees that
 $$ \frac{1}{4} \sum_{e= \{x,y\}} {\rm dem}(e) \|\v_x - \v_y\|^2 = \frac{1}{4} \cdot r \cdot
 \binom{\nfrac{|V|}{r}}{2} \cdot 2 = \frac{D}{2}. $$
Letting $B \defeq \nfrac{D}{2},$  the {SDP} solution is $B$-balanced and its objective
value is at most $\epsilon$. Using Lemma
\ref{lem:existence-cut},    we get a $B'$-balanced cut $(A,
\overline{A})$, $B' \geq  \nfrac{B}{3}$ such that $\sum_{e \in E(A,
\overline{A})} {\rm wt}(e) \leq O(\Gamma \cdot \epsilon)$.

\medskip
\noindent
{\bf Claim}: The cut
$(A,\overline{A})$ must be a $\nfrac{5}{6}$-piecewise balanced cut.

\medskip
\noindent
{\bf Proof of Claim.}
Let $p_i \defeq {\Pr}_{x \in V_i}[A(x)=1].$ The total demand cut by $(A,\overline{A})$ is equal to
$\sum_{i=1}^r p_i (1-p_i) |V_i|^2$. This is at least $B' \geq \nfrac{B}{3}$ since
$(A,\overline{A})$ is $B'$-balanced. Hence,
 $$ \sum_{i=1}^r p_i (1-p_i) \cdot \frac{|V|^2}{r^2} \geq  \frac{1}{6} \ r\cdot
 \binom{\nfrac{|V|}{r}}{2}. $$
Thus, $ \sum _{i=1}^r p_i(1-p_i) \geq \nfrac{r}{12}.$ By Cauchy-Schwarz inequality,
  $$  \E_{i \in _R [r]}  \ \Big| \ \E_{x \in_R V_i}[ A(x) ]
  \Big|
  =
  \frac{1}{r} \sum_{i=1}^r |1-2p_i| \leq \sqrt{\frac{1}{r}\sum_{i=1}^r (1-2p_i)^2}
   = \sqrt{1- \frac{4}{r} \sum_{i=1}^r p_i(1-p_i)} \leq
   \sqrt{\frac{2}{3}} < \frac{5}{6} .$$
Hence, $(A,\overline{A})$ must be a $\nfrac{5}{6}$-piecewise balanced cut.
However, Property (2) of Theorem  \ref{thm:graph-construction}
says that such a cut must cut at least $\epsilon^{t}$ fraction of
edges. This implies that $\Gamma =
\Omega(\nfrac{1}{\epsilon^{1-t}})$. Theorem \ref{cor:gl-disprove} now follows by noting
that $t > \nfrac{1}{2}$ is arbitrary and $n  = |V| \leq
2^{2^{O(\nfrac{1}{\epsilon^3})}}$.

\subsection{Fourier Analysis}\label{sec:prelim-fourier}

Consider the real vector space of all functions $f : \{-1,1\}^n
\mapsto \Real,$ where the addition of two functions is defined to be 
pointwise addition.
 For $f,g:\{-1,1\}^n \mapsto \mathbb{R},$ define the
following inner product:
 $$\langle f, g \rangle _2 \defeq 2^{-n} \sum
_{\x \in \{-1,1\}^n} f(\x)g(\x).$$
For a set $S \subseteq [n],$ define the {\it Fourier character}
$\chi_S(\x) \defeq \prod_{i \in S} x_i.$ It is well-known (and easy to
prove) that the set of all Fourier characters forms an orthonormal
basis with respect to the  above inner product. Hence, every function $f :
\{-1,1\}^n \mapsto \Real$ has a (unique) representation as
 $ f = \sum_{S \subseteq [n]} \widehat{f}_S \chi_S,$  where    $\widehat{f}_S  \defeq \langle f, \chi_S
 \rangle_2$ is the Fourier coefficient of $f$ w.r.t. $S.$ 
The following is a simple but useful fact.
\begin{fct}[Parseval's Identity]\label{fct:parseval} For any
$f:\{-1,1\}^n \mapsto \{-1,1\},$ $\sum_{S \subseteq [n]}
\widehat{f}_S^2=1.$
\end{fct}
The proof of this follows from the following sequence of equalities:
$$ 1= \frac{1}{2^n} \sum_{x \in \{-1,1\}^n} f^2(x)=  \langle f,f \rangle_2 =  \left \langle \sum_{S \subseteq [n]} \widehat{f}_S \chi_S ,  \sum_{T \subseteq [n]} \widehat{f}_T \chi_T \right \rangle_2 =  \sum_{S \subseteq [n]}
\widehat{f}_S^2,$$
where the last equality follows from the orthonormality of the characters $\{\chi_S\}_{S \subseteq [n]}$ with respect to the inner product $\langle \cdot, \cdot \rangle _2.$

\noindent
For the analysis of our {\UG} integrality gap instance presented in Section \ref{sec:UGC}, we need the following notion of an $\ell_p$ norm of a Boolean function. 
 For  $f:\{-1,1\}^n \mapsto \mathbb{R}$ and $p
\geq 1,$ let $$\| f\|_p  \defeq \left( \frac{1}{2^n} \sum _{\x \in
\{-1,1\}^n} |f(\x)|^p\right)^{\nfrac{1}{p}}.$$
We also need to define the so-called Bonami-Beckner operator whose input is a Boolean function $f$ and whose output is again a  Boolean function (which is supposed to be a {\em smoothened} version of $f$).   

\begin{defn}[Hyper-contractive Operator]\label{def:t}
For each $\rho \in [-1,1],$ the Bonami-Beckner operator $T_\rho$ is
a linear operator that maps the space of functions $\{-1,1\}^n
\mapsto \mathbb{R}$ into itself via
$$ T_\rho[f] \defeq \sum _{S \subseteq [n]} \rho^{|S|} \widehat{f}_S\chi_S.$$
\end{defn}

\noindent
The following theorem shows that the Bonami-Beckner operator indeed smoothens $f$: It allows us to upper bound a higher norm of   $T_\rho[f]$ of $f$  with a lower norm of $f$ under certain conditions.  

\begin{theorem}[Bonami-Beckner Inequality \cite{RyanThesis}]\label{thm:bb}
Let $f:\{-1,1\}^n \mapsto \Real$ and $1 < p < q.$ Then
$$ \| T_\rho[f]\|_q \leq \|f\|_p$$ for all ~$0 \leq \rho \leq \left(
\frac{p-1}{q-1} \right)^{\nfrac{1}{2}}.$
\end{theorem}

\medskip
\noindent
The last set of preliminaries are  important for the PCP reduction in Section \ref{sec:inner-bs}.
\begin{defn}[Long Code \cite{BGS}]\label{def:longcode} The Long Code over a domain $[N]$ is indexed by all
$\x \in \{-1,1\}^N$. The Long Code $f$ of an element $j \in [N]$ is
defined to be 
 $ f(\x)  \defeq \chi_{\{j\}}(\x)=x_j,$ for all  $\x=(x_1,\ldots,x_N) \ \in \ \{-1, 1\}^N.$
\end{defn}
\noindent
Thus, a Long Code is simply a Boolean function that is a
dictatorship, i.e., it depends only on one coordinate. In
particular, if $f$ is the Long Code of $j \in [N]$, then
$\widehat{f}_{\{j\}}=1$ and all other Fourier coefficients are zero.

\medskip
\noindent
The following theorem (quantitatively) shows that if a Boolean function is such that its Fourier mass is concentrated on sets of small size, then it must be close to a {\em junta}. In other words, its  Fourier mass on sets with {\em small} Fourier coefficients is {\em small}. 

\begin{theorem}[Bourgain's Junta Theorem \cite{Bourgain}]
\label{thm:bourgain-junta} Fix any $\nfrac{1}{2} < t < 1$. Then,  there
exists a constant $c_t > 0,$ such that, for all positive  integers $k$, for
all $\gamma > 0$ and for all Boolean functions $f:\{-1,1\}^n
\mapsto \{-1,1\},$

$$ \mbox{if} \ \sum_{S \ : \ |S| >  k} \widehat{f}^2_S  < c_t k^{-t}  \quad \mbox{then } \quad
 \sum_{S \ : \ |\widehat{f}_S| \leq \gamma4^{-k^2} }
 \widehat{f}^2_S
  <  \gamma^2.$$
\end{theorem}

\section{The Integrality Gap Instance for  Unique Games}\label{sec:UGC}

In this section, we  present the integrality gap construction for a natural {SDP} relaxation of the {\UG} problem.
We start with defining the {\UG} problem, the
{UGC} of Khot \cite{KhotUCSP} along with the related
preliminaries towards our construction.

\subsection{The Unique Games Problem, its SDP Relaxation and the UGC}
\begin{defn}[{\UG}]
An instance
 $\mathcal{U}=\left(G(V, E), [N], \{ \pi_e\}_{e \in E}, {\rm wt}  \right)$ of {\UG} is
defined as follows: $G(V,E)$ is a
graph with a set of vertices $V$
and a  set of edges $E$.
An
edge $e$ with endpoints  $v$ and $w$ is written as $e\{v,w\}.$
 For every $e\{v,w\} \in E,$ there is a bijection $\pi_e:[N] \mapsto [N]$ and a weight ${\rm wt}(e) \in \mathbb{R}^+.$
The goal is to assign a {\em label} from the set $[N]$
to every vertex of the graph
so as to satisfy the
constraints given by bijective maps $\pi_e.$
A labeling $\lambda:V
\mapsto[N]$ {\em satisfies} an edge $e\{v,w\},$ if $\lambda(v)=
\pi_e (\lambda(w))$.\footnote{We consider the edges to be undirected, but there is an implicit
direction when we write the edge as $e\{v,w\}$ and it is reflected in the bijective constraint that
$\lambda(v)=
\pi_e (\lambda(w))$. The edge could be written in reverse by reversing the bijection.}
Let ${\rm val}(\lambda)$ denote the total weight of the edges satisfied by a labeling $\lambda$:
$$ {\rm val}(\lambda)  \defeq \sum _{e\{v,w\} \in E: \lambda \ \mbox{\rm satisfies} \ e } \ {\rm wt}(e). $$
The optimum {\rm opt}$(\calU)$ of the {\UG}
instance  is defined to be the maximum weight of  edges satisfied by any labeling:
 $${\rm opt} (\calU)  \defeq \max_{\lambda: V \mapsto [N]}  {\rm val}(\lambda).  $$
We assume w.l.o.g that $\sum _{e \in E} {\rm wt} (e)=1$ so that the weights define a probability
distribution over edges. A choice of a random edge refers to an edge chosen from this distribution.
We also assume that
the graph is regular in the sense that
the sum of weights of edges incident on a vertex is the same for all vertices. A choice of a random edge incident
on a vertex $v$ refers to a choice of a random edge conditional on having one endpoint as $v$.
\end{defn}

\begin{conjecture}[{UGC} ~\cite{KhotUCSP}] \label{conj:ugc}
For every pair of  constants $\eta, \zeta > 0$, there exists a sufficiently large
constant $N  = N(\eta,
\zeta)$ such that given a {\UG} instance $\;\calU=\left(G(V, E), [N], \{ \pi_e\}_{e \in E},
{\rm wt}  \right)$, it is {NP}-hard to distinguish whether:
\begin{itemize}
\item {\rm opt}$(\calU) \geq 1-\eta,$ or
\item {\rm opt}$(\calU) \leq \zeta$.
\end{itemize}
\end{conjecture}

\noindent
Consider a
{\UG} instance $\calU=\left(G(V, E), [N], \{ \pi_e\}_{e \in E}, {\rm wt}  \right).$
Khot \cite{KhotUCSP} proposed the {SDP} relaxation in Figure \ref{sdp:ugc} (inspired by a paper of
Feige and Lov\'asz \cite{FL}). Here,
for every $v \in V,$ we associate a set of $N$ orthogonal vectors
$\{ \v_1,\ldots,\v_N\}$. The intention is that if $i_0 \in [N]$ is a label for
vertex $v \in V$, then $\v_{i_0} = \sqrt{N} {\bf 1},$ and $\v_i = {\bf 0}$ for all $i \not= i_0$.
Here, ${\bf 1}$ is some fixed unit vector and ${\bf 0}$ is the zero-vector. However, once we
take the {SDP} relaxation, this may no longer be true and $\{\v_1, \v_2, \ldots, \v_N\}$ could be
any set of orthogonal vectors.

\begin{figure}[htb]
\begin{equation}\label{sdp:ugcobj}
{\rm Maximize} \quad  \sum_{e\{v,w\}\in E} {\rm wt}(e) \cdot \frac{1}{N} \left( \sum_{i=1}^N
\left\langle \v_{\pi_e(i)}, \w_{i} \right\rangle \right) \end{equation}
Subject to
\begin{eqnarray}
\forall \ v \in V & \sum_{i =1}^N \left\langle\v_i, \v_i \right\rangle = N     \label{sdp:ugc1} \\
\forall \ v \in V \ \ \forall \ \ i
\neq j  & \left\langle\v_i, \v_j \right\rangle= 0   \label{sdp:ugc2} \\
 \forall \ v, w \in V \ \ \forall
 \ \ i, j & \left\langle\v_i, \w_j \right\rangle\geq  0 \label{sdp:ugc3} \\
\forall
 \ v, w \in V & \sum_{1\leq i, j\leq N}  \ \left\langle\v_i, \w_j\right\rangle = N     \label{sdp:ugc4}
\end{eqnarray}
 \caption{SDP for {\UG}}\label{sdp:ugc}
\end{figure}

\subsection*{The Noisy Hypercube  and an Overview of the Integrality Gap Instance}
With a {\UG} instance with $N$ labels, one can associate a related graph called the {\it label extended graph}. It
turns out that the optimum of the {\UG} instance is closely related to the expansion of {\it small} sets, namely those of relative size
$\nfrac{1}{N}$, in the label extended graph. In particular, if all sets of size $\nfrac{1}{N}$ in the label extended
graph have a near-full expansion, then the optimum of the {\UG} instance is low.
Our integrality gap construction  starts with a so-called {\it noisy hypercube graph} on vertex set $\{-1,1\}^N$
and obtain a {\UG} instance from it so that the former is precisely the label extended graph of the latter. The fact that
the {\UG} instance has low optimum then follows directly from the observation that the noisy hypercube
graph is a small set expander (its proof via the Bonami-Beckner inequality was pointed out to us
by Ryan O'Donnell).
The SDP solution for the {\UG} instance is constructed using the vertices of the hypercube
thought of as vectors in $\mathbb{R}^N$.

\begin{remark} The idea of the label extended graph and the implication that the small set expansion in the
label extended graph implies low optimum for the {\UG} instance were implicit in the conference version of this paper  \cite{KhotV05}.
We choose to make this more explicit here for the ease of presentation as well as in light of
recent works that we briefly mention. Raghavendra and Steurer recently proposed the Small Set Expansion Conjecture \cite{RS-SSE} and showed that it implies the UGC. The former states that for every constant $\epsilon > 0$, there exists a constant $\delta > 0$ such that
given an $n$-vertex graph that has a {\it small} non-expanding set,
i.e., of size $\delta n$ and with edge expansion at most $\epsilon$, it is NP-hard to find a set of size (roughly) $\delta n$
that is even somewhat non-expanding, i.e., with expansion at most $1-\epsilon$. The SSE Conjecture has led to many interesting works including a new
algorithm for {\UG} by Arora, Barak and Steurer \cite{ABS} and the construction of the {\em short code} \cite{BarakGHMRS12}.
\end{remark}

\begin{defn} Given a {\UG} instance $\calU=\left(G(V, E), [N], \{ \pi_e\}_{e \in E}, {\rm wt}  \right),$ the corresponding
label extended graph $G'(V',E', {\rm wt}')$ is defined  as follows:
\begin{itemize}
\item $V' = V \times [N]$.
\item $\forall e\{v,w\} \in E, i \in [N]$, we let $e' \{ (v, \pi_e(i)), (w,i) \} \in E'$ and ${\rm wt}'(e') = {\rm wt}(e)$.
\end{itemize}
Note that $\sum_{e' \in E'} {\rm wt}'(e') = N$.
\end{defn}

\noindent
It is helpful to view the label extended graph as being obtained from the {\UG} graph by replacing every vertex $v$ by a
group of $N$ vertices representing labels to $v$
and replacing every edge $e\{v,w\}$ by an {\it edge-bundle} of $N$ edges that form a perfect matching between the two
groups and capture the bijective constraint $\pi_e$.

\smallskip
The expansion $\Phi(S')$   of a set $S' \subseteq V'$ in the label extended graph
is defined to be  the probability
of leaving $S'$ when a random vertex in $S'$ and then a random edge leaving that vertex (w.r.t. the weights ${\rm wt}'$)
is chosen. Note that $\Phi(S') \in [0,1]$.
Any labeling $\lambda: V \mapsto [N]$ to a {\UG} instance corresponds to the set $S'_\lambda \subseteq V'$
as follows:
 $$S'_\lambda  \defeq \{ (v, \lambda(v)) \ |  \ v \in V \}. $$
An easy observation is that the (weighted) fraction of edges satisfied by a labeling $\lambda$ is related to
the expansion of the set  $S'_\lambda$:
\begin{equation} \label{eqn:ug-exp}
  {\rm val}(\lambda) =  1 - \Phi(S'_\lambda).
\end{equation}
Here is a quick proof of the above equality.
Pick a random vertex $(v, \lambda(v))$ in $S'_\lambda$ by
choosing a random vertex $v \in V$. Choosing a random edge incident on $(v, \lambda(v))$
(w.r.t. ${\rm wt}'$) amounts to choosing a random edge $e\{v,w\}$ incident on $v$ (w.r.t. ${\rm wt}$) and outputting
$\{ (v, \lambda(v)), (w, \pi_e^{-1} (\lambda(v))) \}$. The expansion of $S'_\lambda$ is now related to
the event that $(w, \pi_e^{-1} (\lambda(v))) \in  S'_\lambda$ which is same as the event that
$\pi_e^{-1} (\lambda(v)) = \lambda(w)$ which is same as the event that $\lambda$ satisfies the edge $e\{v,w\}$.

\smallskip
As remarked before, our construction starts with the noisy hypercube graph and uses the fact that the graph is a small
set expander. A natural way to describe this graph is by describing one step of the random walk on it (which then naturally
leads to edge-weights with unit total weight).
\begin{defn}\label{def:nh} The noisy hypercube graph $H$ with parameters $N$ and $0 < \eta < \nfrac{1}{2}$ has
\begin{itemize}
\item the vertex set $\{-1,1\}^N$ with uniform distribution and
\item for any vertex ${x} \in \{-1,1\}^N$, choosing a random edge $({x}, {y})$ incident on ${x}$  amounts to
      flipping every bit of ${x}$ with probability $\eta$ independently and letting ${y}$ to be the string so obtained.
\end{itemize}
\end{defn}

\begin{lemma} \label{l:noisy-hyp}
Let $H$ be the noisy hypercube with parameters $N$ and $\eta$ and $S \subseteq \{-1,1\}^N$ be a set of relative
size $\nfrac{1}{N}$. Then $1-\Phi(S) \leq  \nfrac{1}{N^{\eta+\eta^2}}$.
\end{lemma}
\begin{proof} Let $f: \{-1,1\}^N \mapsto \{0,1\}$ be the indicator function of the set $S$ so that $\| f\|_p^p =
\nfrac{1}{N}$ for any $1 \leq p < \infty$.  An application of Bonami-Beckner inequality gives
 (the probability is taken over choice of a random vertex ${x}$ and a random
edge $({x}, {y})$ incident on it)
\begin{eqnarray*}
1 - \Phi(S) & =  &  \Pr \left[ {y} \in S  \  |  \ {x} \in S \right]
             =    \frac{ \Pr \left[ {x} \in S, \  {y} \in S \right] }{ \Pr  \left[ {x} \in S\right]  }
             =   N \cdot  \Pr \left[ {x} \in S, \  {y} \in S \right] \\
            & = & N \cdot  { \E}_{{x}, {y}} [ f({x}) f({y}) ]
             =  N \cdot  \sum_{ \alpha \subseteq [N] }  \widehat{f}_\alpha^2 (1-2\eta)^{|\alpha|}
             \stackrel{{\rm Def. \ref{def:t}}}{=}   N \cdot \| T_{\sqrt{1-2\eta}} f  \|_2^2  \\
            & \stackrel{{\rm Thm.} \ref{thm:bb}}{\leq} & N \cdot \| f  \|_{2-2\eta}^2 = N \cdot \left( \frac{1}{N} \right)^{\nfrac{2}{(2-2\eta)}} \leq
            N \cdot \frac{1}{N^{1+\eta+\eta^2}} = \frac{1}{ N^{\eta+\eta^2}}.
\end{eqnarray*}
\end{proof}

\noindent
Call an edge $({x}, {y})$ of the noisy hypercube {\it typical} if the Hamming distance between
${x}$ and ${y}$ is close to $\eta N$, say between $\frac{\eta}{2}N$ and $2\eta N$. By the Chernoff bound,
the (weighted) fraction of edges which are not {\it typical} is at most $2^{-\Omega(\eta N)}$ which is negligible
in our context. We delete all these edges (mainly for the ease of presentation) and observe that the conclusion of Lemma
\ref{l:noisy-hyp} still holds with the bound $1-\Phi(S) \leq \nfrac{1}{N^{\eta}}$.
The weights of the edges change
slightly, due to a re-normalization to preserve the unit total weight, but we ignore this issue.

\smallskip
We are now ready to construct an integrality gap instance for the SDP in Figure \ref{sdp:ugc}. To be precise,
for parameters $N$ and $\eta$,
we  construct an instance $\calU= \left(G(V, E), [N], \{ \pi_e\}_{e \in E}, {\rm wt}  \right)$ of
{\UG} such that
\begin{itemize}
\item (Soundness) ${\rm opt}(\calU) \leq \nfrac{1}{N^\eta}$ and
\item (Completeness) There is an SDP solution with objective value at least $1-9\eta$.
\end{itemize}
This construction is  used later to construct integrality gap instances for cut problems.
As mentioned earlier, the {\UG} instance is constructed precisely so that the noisy hypercube graph happens to be its
label extended graph and then the soundness guarantee follows from Lemma \ref{l:noisy-hyp}.  The vertex
set of the noisy hypercube graph is $\{-1,1\}^N$ where $N = 2^k$. It is convenient
for us to identify a point in $\{-1,1\}^N$ as a Boolean function $f: \{-1,1\}^k \mapsto \{-1,1\}$. We
describe the construction formally now.

\subsection{The Integrality Gap Instance}

Let $\calF$ denote the family of all Boolean functions on $\{-1,1\}^k.$  For $f,g \in \calF,$ define the product $f  g$ as $$(f g)(\x) \defeq f(\x)g(\x).$$ Consider the equivalence relation $\equiv$ on $\calF$ defined to be
 $ f \equiv g$  if and only if there is an $S \subseteq [k],$ such that $f=g  \chi_S$ (recall that
 $\chi_S$ is the Fourier character function).  This
 relation partitions $\calF$ into equivalence classes $\calP_1,\ldots, \calP_m$, each class containing exactly $N=2^k$
 functions.
We denote by $[\calP_i]$ one arbitrarily chosen function in  $\calP_i$ as its representative.
Thus, by definition,
$$ \calP_i=\{ [\calP_i]\chi_S \ | \ S \subseteq [k]\}.$$
It follows from the orthogonality of the characters $\{\chi_S\}_{S \subseteq [k]}$, that all the functions in any class
are also mutually orthogonal.
Further, for a function $f \in \calF,$ let $\calP(f)$ denote the class $\calP_i$ in which $f$ belongs.

\smallskip
Let $\bmu \in_\eta \calF$ denote a random {\it perturbation} function on $\{-1,1\}^k$ where
for every $\x \in \{-1,1\}^k,$ independently, $\bmu(\x)=1$ with probability $1-\eta,$ and $-1$ with probability $\eta.$
Let $H$ be the noisy hypercube graph: It is a graph with vertex set $\calF$ and
for Boolean functions $f,g \in \calF,$ the weight of the edge $\{f,g\}$ is defined as follows:
$$ {\rm wt}'(\{f,g\}) \defeq   \Pr _{h \in  \calF,  \; \bmu \in_\eta \calF}
\left[ ((f=h) \wedge (g= h \bmu) )  \vee ((f = h\bmu) \wedge (g = h))  \right],$$
where $h$ is a uniformly random function and $\bmu$ is a random perturbation function.
Note that the
sum of weights over all (undirected) edges is $1$.
Moreover, for any $S \subseteq [k],$ we have ${\rm wt}' (\{f,g\})= {\rm wt}' (\{f\chi_S,g \chi_S\}).$ We delete
all edges $\{f,g\}$ such that the Hamming distance between $f$ and $g$ is outside the range
$[\frac{\eta}{2}N, 2\eta N]$ without really affecting anything as observed before.

\smallskip
The {\UG} instance  $\calU=\left(G(V, E), [N], \{ \pi_e\}_{e \in E}, {\rm wt}  \right)$
is now obtained by taking the
noisy hypercube graph $H$ as above with a {\it grouping} of its vertices into classes $\calP_1,\ldots,\calP_m$.
The edges of $H$ are grouped neatly into edge-bundles: A typical bundle is a set of
$N$ edges between $\calP_i$ and  $\calP_j$, all with
the same weight, and forming a perfect matching between the $N$ vertices in each group.
With this grouping
in mind, the graph can now be naturally thought of as a label extended graph. The {\UG} instance is
obtained by thinking of each class $\calP_i$ as a (super-)vertex, each function $f \in \calP_i$ as a
potential label to it, and the edge bundle between $\calP_i, \calP_j$ as defining the bijective
constraint between them.  Here is a formal (somewhat tedious) description.

\smallskip
The {\UG}
graph $G(V,E)$ is defined as follows.  The set of vertices is
$V \defeq\{\calP_1,\ldots,\calP_m\}$ as above.   For every $f,g \in \calF$ with Hamming
distance in the range $[\frac{\eta}{2}N, 2\eta N]$,   there
is an edge in $E$ between the vertices $\calP(f)$ and $\calP(g)$
with weight $${\rm wt}(\{ \calP(f), \calP(g)  \}) \defeq N \cdot {\rm wt}'( \{f,g\})$$ (the factor
of $N$ reflects the fact that there are $N$ pairs of functions that define the same edge).
The set of labels for the
{\UG} instance is $2^{[k]} \defeq \{S: S \subseteq [k] \}$, i.e.,
the set of labels $[N]$ is identified with the set $2^{[k]}$ (and by design
$N=2^k$).
Note that $f = [\calP_i] \chi_S$ and $g = [\calP_j] \chi_T$ for some sets $S, T \subseteq [k]$.
The bijection $\pi_e,$ for the edge $e\{\calP_i, \calP_j\}$,
can now be
defined:
$$ \pi_{e}( T \star U)  \defeq S \star U, \ \ \ \forall  \ U \subseteq [k].$$
Here, $ \star$ is the symmetric difference operator on sets.     Note that $\pi_e : 2^{[k]} \mapsto
2^{[k]}$ is a permutation on the set of allowed labels. An alternate view is that the potential labels to
class $\calP_i$ are really the functions in that class and for the edge defined by a pair $f \in \calP_i$
and $g \in \calP_j$ as above, $\pi_e$ designates $(f\chi_U, g\chi_U)$ as a matching pairs of labels for all
$U \subseteq [k]$. We emphasize that every matching pair of labels corresponds to a pair of functions with
Hamming distance in $[\frac{\eta}{2}N, 2\eta N]$.

\subsection*{Soundness: No Good Labeling}
Using Lemma \ref{l:noisy-hyp} and Equation \eqref{eqn:ug-exp}, i.e., the connection between the optimum of {\UG}  and the small set expansion of the label extended graph, it follows immediately that any labeling to the {\UG} instance described above
achieves an objective of at most $\nfrac{1}{N^\eta}.$

\subsection*{Completeness: A Good SDP Solution}\label{sec:ULC-complete}

 For $f \in \calF,$ let $\u_f$ denote the unit vector (w.r.t. the $\ell_2$ norm) corresponding to the truth-table of $f.$
Formally, indexing the vector $\u_f$ with coordinates $\x \in \{-1,1\}^k,$ 
$$(\u_f)_\x \defeq \frac{f(\x)}{\sqrt{N}}.$$

\smallskip
\noindent
Recall that in the SDP relaxation of {\UG} (Figure \ref{sdp:ugc}), for every vertex
in $V,$ we need to assign a set of orthogonal vectors. For every vertex $\calP_i \in V$, we choose a function
$f \in \calP_i$ arbitrarily, and
with $\calP_i,$ we associate the set of vectors
$\left\{\u_{f \chi_S} ^{\otimes 2}\right\}_{S \subseteq [k]}.$ The
following facts are easily verified:

\begin{enumerate}
\item  $\sum_{S \subseteq [k]} \left\langle\u_{f \chi_S} ^{\otimes 2}, \u_{f \chi_S} ^{\otimes 2} \right\rangle =  \sum_{S \subseteq [k]} \left\langle \u_{f \chi_S},\u_{f \chi_S}  \right\rangle^2 = N.$

\item For $S \neq T \subseteq [k],$ $\left\langle \u_{f\chi_S}^{\otimes 2},\u_{f \chi_T}^{\otimes 2} \right\rangle =  \left\langle \u_{f \chi_S},\u_{f \chi_T}  \right\rangle^2=  \left\langle \u_{\chi_S},\u_{ \chi_T}  \right\rangle^2=0.$

\item For $f,g \in \calF$ and $S,T \subseteq [k],$  $\left\langle \u_{f\chi_S}^{\otimes 2},\u_{g \chi_T}^{\otimes 2} \right\rangle =  \left\langle \u_{f \chi_S},\u_{g \chi_T}  \right\rangle^2 \geq 0.$

\item For $f \in \calP_i,$  $g\in \calP_j$ for $i \neq j,$
$$\sum_{S,T \subseteq [k]}  \left\langle\u_{f \chi_S} ^{\otimes 2}, \u_{g \chi_T} ^{\otimes 2}
\right\rangle=\sum_{S,T \subseteq [k]} \left\langle \u_{f
\chi_S},\u_{g \chi_T} \right\rangle^2= \sum_{T \subseteq [k]}
\left\|\u_{g \chi_T} \right\|^2=N.$$ Here, the second last equality
follows from the fact that, for any $f \in \calF,$  $\{\u_{f
\chi_S}\}_{S \subseteq[k]}$ forms an orthonormal  basis for
$\mathbb{R}^N.$
\end{enumerate}

\noindent
Hence, all the conditions \eqref{sdp:ugc1}-\eqref{sdp:ugc4}  of the SDP are
satisfied. Next, we show that this vector assignment has an
objective at least $1-9\eta.$
Consider any {\UG} edge
defined by a pair $f,g$ with Hamming distance in the range $[\frac{\eta}{2}N, 2\eta N]$. For any $S \subseteq [k]$,
note that the same edge is defined by the pair $f\chi_S, g \chi_S$ with the same Hamming distance and

$$\left\langle \u_{f \chi_S}^{\otimes 2},\u_{g \chi_S}^{\otimes 2} \right\rangle
 =  \left\langle \u_{f \chi_S},\u_{g \chi_S} \right\rangle^2    \geq (1-4\eta)^2  \geq 1-8\eta.$$
Since the pairs $(f\chi_S, g\chi_S)$ are precisely the
  matching pairs of labels for the {\UG} constraint, it follows that
the objective of this SDP solution is at least
$1-9\eta$ (accounting possibly for the {\it non-typical} pairs $f,g$ with
Hamming distance outside of range $[\frac{\eta}{2}N, 2\eta N]$ that were deleted and
ignored throughout).
Finally, note that since all the vectors have coordinates either $1$ or $-1$ (up to a normalization factor),
any three vectors $u,v,w$ among those described above satisfy the triangle inequality:
$$ 1+ \langle u,v \rangle \geq \langle v,w \rangle + \langle u, w \rangle.$$

\subsection*{Summarizing and Abstracting the Unique Games Instance}
 For future reference, we summarize and abstract out the key
properties of the integrality gap construction in the theorem below.  Therein,
for every vertex $v\in V$ of the {\UG}
instance, there is an associated  set
of vectors $\{\v_i^{\otimes 2} \}_{i \in [N]}.$ Moreover, $[N]$ has
a group structure with addition operator $\oplus$ (the group being
$\mathbb{F}_2^k$ and $i \in [N]$ identified with the corresponding group element).
Additionally, we keep track of the parameter $\eta$ and denote the instance by $\calU_\eta.$

\begin{theorem}\label{thm:ulc-sdp}
For any  $0< \eta <\nfrac{1}{2}$ and any integer $N$ that is a power of $2$,
there is a {\UG} instance $\calU_\eta=\left(G(V, E), [N], \{ \pi_e\}_{e \in E}, {\rm wt}\right)$ along with
a set of vectors
$\{\v_i^{\otimes 2} \}_{i \in [N]}$ for every vertex such that:
\begin{enumerate}
\item $|V| = \tilde{n} = \nfrac{2^N}{N}$  and  ${\rm opt}(\calU_\eta) \leq \log^{-\eta} \tilde{n}$.
\item {\bf Orthonormal Basis} \\
The set of vectors $\{\v_{i}\}_{i \in [N]}$ forms an
orthonormal basis for the space $\mathbb{R}^{N}.$  Hence, for any
vector $\w \in \mathbb{R}^N,$ $\|\w\|^2= \sum_{i \in [N] }
 \langle \w, \v_i \rangle^2.$
\item {\bf Triangle Inequality} \\
For any $u,v,w \in V,$ and any $i,j,\ell \in [N],$ $ 1+\langle
\u_i , \v_j \rangle \geq \langle \u_i,\w_\ell \rangle +  \langle
\v_j,\w_\ell \rangle.$

\item {\bf Matching Property} \\
For any $v,w \in V,$ and $i,j,\ell \in [N],$
$\langle \v_i,\w_j \rangle =\langle \v_{i \oplus \ell},\w_{j \oplus \ell} \rangle.$

\item {\bf Closeness Property}  \\
For any $e\{v,w\} \in  E$,
there are  $i_0, j_0 \in [N]$
such that $\langle \v_{i_0}, \w_{j_0} \rangle \geq 1-4\eta.$  Moreover,
if $\pi_e$ is the bijection corresponding to this edge, then
$i_0 \oplus \ell = \pi_e(j_0\oplus \ell)  $ for all $\ell \in [N]$.

\end{enumerate}
\end{theorem}

\section{A PCP Reduction from Unique Games to Balanced Edge-Separator}\label{sec:inner-bs}

This section presents the reduction from {\UG} to non-uniform {\BS} which underlies the proof of Theorem \ref{thm:bs-main}.
Remark \ref{rem:sc-bs} implies that if non-uniform {\BS} is hard to approximate within a factor of $C,$ then so is non-uniform {\SC}
up to a factor $\Omega(C)$.
Hence, Theorem \ref{thm:bs-main} can be strengthened as follows.

\begin{theorem}\label{thm:hardness-sc}
Assuming the {UGC}, it is {NP}-hard to approximate (non-uniform versions of) {\BS} and {\SC} to within any constant factor.
\end{theorem}

\noindent
We present the reduction and the proof of this theorem, modulo the soundness proof of the PCP reduction. The soundness proof is (by now) standard and relegated to Appendix \ref{sec:app-pcp}.
The reduction underlying the proof of this theorem is used in the construction of the integrality gap  for {\BS} presented in Section \ref{sec:INTGAP}.

\subsection*{Overview of the Reduction} The reduction  starts with a {\UG} instance $\calU=\left(G(V, E), [N], \{ \pi_e\}_{e \in E}, {\rm wt}\right)$.  Each vertex $v\in V$ is replaced with a {\em block} of vertices $\{(v,x): x \in \{-1,1\}^N \}.$  The reduction has a parameter $\epsilon$ which is to be thought of as a small constant. For each edge $e\{v,w\}$ in $\calU,$ a {\em bundle} of weighted edges are put between the two corresponding blocks of vertices taking into account the permutation $\pi_e$ corresponding to that edge. The weight of the edge between $(v,x)$ and $(w,y)$ is equal to the product of the weight of the edge $e\{v,w\}$ and the probability that, if we flip each bit of $x$ independently with probability $\epsilon,$ we obtain $ y \circ \pi_e.$
 Here $y \circ \pi_e$ is the reordering of the coordinates of $y$ as dictated by $\pi_e$; formally, $(y \circ \pi_e)_i=y_{\pi_e(i)}$ for all $i \in [N].$

\medskip
Note that if we contract the vertices of the two hypercubes after identifying the coordinates according to $\pi_e,$ we obtain exactly the noisy hypercube introduced in Definition \ref{def:nh}. To complete the reduction, we need to specify the demand pairs. For reasons that will become clear in a bit, any pair of vertices  in the same block is set to have   demand one and the remaining pairs have  demand zero.

\medskip
Our reduction has  the property that if the {\UG} instance has a good labeling then there is a cut that cuts a constant fraction of the demand pairs and the weight of the edges crossing the cut is small. This is by construction: If the {\UG} instance $\calU$ has a good labeling, i.e., a  $\lambda: V \mapsto [N]$ which satisfies  at least a $1-\epsilon$ fraction of the constraints of $\calU,$ then we consider the cut in the reduced graph whose one side consists of the vertices  $(v,x)$ such that $x_{\lambda(v)}=1$ and the other side with vertices $(v,x)$ such that $x_{\lambda(v)}=-1.$ It is easy to see  that the weight of the edges that cross this cut is $1-(1-\epsilon)(1-\epsilon)=O(\epsilon).$ Moreover, the number of demand pairs cut is half that of the total demand pairs as the cut described above cuts each hypercube along a coordinate into two equal parts.  This is the completeness of the reduction.

\medskip
For soundness, we show that  if every labeling of the {\UG} instance satisfies a negligible  (as a function of $\epsilon$) fraction of the constraints,  any cut in the reduced graph that cuts a constant fraction of demand pairs must have about $\sqrt{\epsilon}\gg \epsilon$ weight of edges crossing it. Since the reduction is local in the sense that it replaces each vertex in $\calU$ by a set of vertices, and each edge in $\calU$ by a bundle of edges between the corresponding sets, the weighted graph obtained by applying this reduction on $\calU$ inherits  connectivity properties of $\calU.$ For instance, if $\calU$ is disconnected, then there is a cut in the reduced graph which has no edges crossing it. Such a cut, however, puts each hypercube entirely on one side of the cut or the other, thus, cutting {\em no} demand pair. Hence, the way we have enforced demands essentially ensures that each cut in the reduced graph that cuts a constant fraction of demand pairs cuts most of the hypercubes into two roughly equal parts. Hence, for each vertex $v$ in $\calU$ we can look at the restriction of this cut to the corresponding hypercube and assign to $v$ the label corresponding to the dimension of the hypercube which is the most {\em correlated} with the cut restricted to that hypercube. Since $\calU$ does not have a good labeling, this strategy of converting a cut in the reduced graph to a labeling for $\calU$ should not be good. Hence, one can deduce that, for any cut that cuts a constant fraction of the demand in the reduced graph, its restrictions to most hypercubes must not be well-correlated to any coordinate cut.  This is where Bourgain's Junta theorem (Theorem \ref{thm:bourgain-junta}) comes in. It essentially implies that such a cut must be close to a majority cut in most hypercubes.
This  allows us to deduce that such a cut has at least  $\sqrt{\epsilon}$ weight edges crossing it, giving us the hardness of approximation ratio $\approx \nfrac{\sqrt{\epsilon}}{\epsilon}$ which can be made larger than any constant by choosing $\epsilon$ small enough.

\medskip
We now describe the reduction formally. Here, it is instructive to break the reduction into two parts: The first consists of presenting a PCP verifier for {\UG} and the second step involves  translating the PCP verifier into a {\BS} instance. The completeness and the soundness of this verifier give us the proof of Theorem \ref{thm:hardness-sc}.

\subsection{The PCP Verifier}

For $\epsilon \in(0,1),$  we present a  PCP verifier  which
  given a {\UG} instance $\ {\calU}=(G(V,E),[N],
\{\pi_{e}\}_{e\in E})$ decides whether ${\rm opt}(\calU) \sim 1$ or ${\rm opt}(\calU) \sim 0.$
The verifier $V_\epsilon$
expects, as a proof, the Long Code (see Definition \ref{def:longcode}) of the label of every vertex $v \in V.$ Formally, a proof  $\Pi$ is $\{A^v\}_{v \in V},$ where each $A^v: \{-1,1\}^N \mapsto \{-1,1\}$ is the supposed Long Code of the label of $v.$
The actions of $V_\epsilon$ on $\Pi$ are as follows.
\begin{enumerate}
\item Pick $e\{v,w\} \in E$ with probability ${\rm wt}(e)$.
\item  Pick a random $\x \in _{\nfrac{1}{2}} \{-1,1\}^N$
and $\bmu \in_{\epsilon}  \{-1,1\}^N$.

\item Let $\pi_{e} : [N] \mapsto [N]$ be the bijection
corresponding to $e\{v,w\}.$ {\bf Accept} if and only if
$$ A^v(\x) = A^w((\x\bmu)\circ\pi_e).$$
\end{enumerate}

\noindent
The completeness of verifier is easy and we provide a proof here.

\begin{lemma}[Completeness] \label{thm:bs-c}
For every $\epsilon \in (0,1),$ if {\rm opt}$(\calU) \geq 1-\eta,$
 there is a proof $\; \Pi$ such that $$\Pr\left[
V_\epsilon\;\rm{accepts} \; \Pi \right] \geq
(1-\eta)(1-\epsilon).$$   Moreover, every table $A^v$ in $\Pi$
is balanced, i.e., exactly half of its entries are $+1$ and the
rest  are $-1$.
\end{lemma}
\begin{proof}
Since {\rm opt}$(\calU) \geq 1-\eta,$ there is a labeling
$\lambda$ for which the total weight of the edges satisfied  is
at least $1-\eta.$ Hence, if we pick an edge $e\{v,w\}$ with
probability ${\rm wt}(e),$ with probability at least $1-\eta,$ we
have  $\lambda(v)=\pi_e(\lambda(w)).$ Let the proof consist of
Long Codes of the labels assigned by $\lambda$ to the vertices.
With probability $1-\epsilon,$ we have ${\mu}_{\lambda(v)}=1.$
Hence, with probability at least $(1-\eta)(1-\epsilon),$
$$A^v(\x)=\x_{\lambda(v)}=(\x {\mu})_{\pi_e(\lambda(w))}=A^w((\x{\mu}) \circ \pi_e).$$
Noting that a Long Code is balanced, this completes the proof. \end{proof}

\noindent
The soundness of the reduction involves more work and, since \cite{KhotUCSP, KhotV05},
has become standard. We state the result here and the proof appears in Appendix \ref{sec:app-pcp}.
We say that a proof $\Pi=\{A^v\}_{v\in V}$ is
$\theta$-piecewise balanced if
$$ \E _{v}\left[ |\widehat{A}_\emptyset ^v|\right] \leq \theta.$$
Here, $\widehat{A}_\emptyset ^v$ is the Fourier coefficient corresponding to the empty set of the Boolean function $A^v$ and the expectation is over
a uniformly random vertex $v \in V$.

\begin{lemma}[Soundness]\label{lem:bs-s} For every $t \in
(\nfrac{1}{2}, 1)$, there exists a constant $b_t > 0$ such that the
following holds: Let $\epsilon > 0$ be sufficiently small and let
$\: \calU$  be an instance of {\UG} with ${\rm opt}(\calU) <
2^{-O(\nfrac{1}{\epsilon^2})}.$ Then, for every $\nfrac{5}{6}$-piecewise balanced
proof $\Pi,$  $$\; \Pr\left[ V_\epsilon\;\rm{accepts} \;
\Pi \right] < 1-b_t\epsilon^t.$$
\end{lemma}

\subsection{From the PCP Verifier to a Balanced Edge-Separator Instance}\label{sec:reduction}
The reduction from the {PCP} verifier  to an instance ${\cal I}_\epsilon$ of
 non-uniform {\BS}  is as follows. Replace the bits in the proof by vertices and replace every ($2$-query) {PCP} test by an edge
of the graph.
 The weight of the edge is equal to the probability that
the test is performed by the {PCP} verifier.
Formally, we start with a {\UG} instance $\calU=\left(G(V, E), [N], \{ \pi_e\}_{e \in E}, {\rm wt}  \right),$
and replace each vertex  $v \in V$  by a {\em block} of vertices $(v,\x)$ for each $\x \in \{-1,1\}^N.$
For an edge $e\{v,w\}\in E,$  there is an edge in  ${\cal I}_\epsilon$ between $(v,\x)$ and $(w,\y),$ with  weight
$$ {\rm wt}(e)\cdot \Pr_{\stackrel{\x' \in_{\nfrac{1}{2}} \{-1,1\}^N} {\bmu \in _\epsilon \{-1,1\}^N}}\left[\left(\x=\x'\right) \wedge \left( \y=\x'{\bmu} \circ \pi_e \right) \right].$$
This is exactly the probability that $V_\epsilon$
picks the edge $e\{v,w\},$ and decides to look at the $\x$-th (resp. $\y$-th) coordinate in the Long Code of the label of $v$ (resp. $w$).

\medskip
 The demand function {\rm dem}$(\cdot)$ is $1$ for any edge between vertices in the same block, and $0$ otherwise. Let $B \defeq \frac{1}{2} \cdot |V| \cdot  \binom{2^N}{2}$ be half of the total demand.

\medskip
 Assuming the {UGC}, for any $\eta, \zeta > 0$, for a sufficiently large $N$,  it is
{NP}-hard to determine whether an instance $\calU$ of {\UG}
has ${\rm opt}(\calU) \geq 1-\eta$ or ${\rm opt}(\calU) \leq
\zeta$. We choose $\eta = \epsilon$ and $\zeta \leq
2^{-O(\nfrac{1}{\epsilon^2})}$ so that 
\begin{enumerate}
\item[(a)] when {\rm opt}$(\calU) \geq
1-\eta,$ there is a (piecewise balanced)  proof that the verifier
accepts with probability at least $1-2\epsilon$ and 
\item [(b)] when {\rm opt}$(\calU) \leq \zeta$,  the verifier does not accept any
$\nfrac{5}{6}$-piecewise balanced proof with probability more than
$1-b_t\epsilon^t.$
\end{enumerate}
 Note that $b_t$ is defined as in the statement of Lemma  \ref{lem:bs-s}.

\medskip
Suppose that {\rm opt}($\calU$) $\geq 1-\eta.$ Let $\lambda$ be a
labeling that achieves the optimum. Consider the partition
$(S,\overline{S})$ in   ${\cal I}_\epsilon$ such that $S$
consists of all vertices $(v,\x)$ with the property that the Long
Code of $\lambda(v)$ evaluated at $\x$ is $+1.$ Clearly, the demands
cut by this partition is exactly equal to $B$. Moreover, it
follows from Lemma \ref{thm:bs-c} that this partition cuts edges
with weight  at most $\eta+\epsilon=2\epsilon$.

\medskip
 Now, suppose that {\rm opt}($\calU$) $\leq \zeta.$ Then,
it follows from Lemma \ref{lem:bs-s}, that any $B'$-balanced
partition, with $B' \geq \nfrac{B}{3},$ cuts at least
$b_t\epsilon^t$ fraction of the edges.  This is due to the
following: Any partition $(S,\overline{S})$ in ${\cal
I}_\epsilon$ corresponds to a proof $\Pi$ in which we let
the (supposed) Long Code of the label of $v$ to be $+1$ at the
point $\x$ if $(v,\x) \in S,$ and  $-1$ otherwise. Since $B' \geq \nfrac{B}{3},$ as in the proof of Theorem
\ref{thm:graph-construction}, $\Pi$ is $\nfrac{5}{6}$-piecewise balanced
and we apply Lemma \ref{lem:bs-s}.

\medskip
Thus, we get a hardness factor of $\Omega\left(\nfrac{1}{\epsilon}^{1-t}\right)$ for {\BS} and, hence, by Remark \ref{rem:sc-bs}, for
{\SC} as well. This completes the proof of Theorem \ref{thm:hardness-sc}.

\section{The Integrality Gap Instance for Balanced Edge-Separator}\label{sec:INTGAP}
In this section, we describe the integrality gap instance for
{\BS} along with its {SDP} solution and prove Theorem \ref{thm:graph-construction}. As pointed out in Section \ref{sec:sc-bs}, this
 also implies an integrality gap for non-uniform {\SC}. The following is, thus, a strengthening of Theorem \ref{thm:bs-main}.
\begin{theorem}\label{thm:bs-restate} Non-uniform versions of {\SC} and {\BS} have an integrality gap of
at least $(\log\log n)^{\nfrac{1}{6}-\delta},$ where $\delta > 0$ is
arbitrary. The integrality gaps hold for standard {SDP}s with
triangle inequality constraints.
\end{theorem}
\noindent
We present a proof of this theorem (by proving Theorem \ref{thm:graph-construction}). The fact that our SDP solution satisfies the triangle inequality constraints  relies on a technical lemma whose proof is via an extensive case analysis and is not very illuminating, hence, relegated to Appendix \ref{sec:sdp-triangle}.

\subsection*{Overview of the Integrality Gap Instance}
The integrality gap instance for non-uniform {\BS} has two parts: A (weighted) graph $(V^*,E^*)$ on $n$ vertices along with demand pairs and a unit vector $\V_u$ for each vertex $u \in V^*.$ The integrality gap instance is parameterized by $\epsilon>0$
and  $\mathcal{I}_\epsilon$  denotes the instance. We show that 
\begin{enumerate}
\item every cut in $V^*$ that cuts a constant fraction of the demand pairs must have at least $\sqrt{\epsilon}$ fraction of edges crossing it and that 
\item  the set of vectors $\{\V_u\}_{u \in V^*}$ satisfy the  constraints in the SDP in Figure \ref{fig:sdpbs} and have an objective value $O(\epsilon),$ thus, giving us an integrality gap of $\Omega (\sqrt{\epsilon}).$ 
\end{enumerate}
The smallest value $\epsilon$ can take turns out to be $(\log \log n)^{\nfrac{-1}{3}}$, giving us the lower bound $\Omega ((\log \log n)^{\nfrac{-1}{6}}).$

\medskip
The graph in $\mathcal{I}_\epsilon$ is obtained by applying the reduction from {\UG} to {\BS} presented in Section \ref{sec:inner-bs} to the {\UG} integrality gap instance $\calU_\eta$ from Section \ref{sec:UGC}, see  Theorem \ref{thm:ulc-sdp} for a summary.
Recall that $\calU_\eta$ consists of the constraint graph $(G(V,E),[N], \{\pi_e\}_{e \in E}, {\rm wt})$ and a set of vectors $\{v_i ^{\otimes 2}\}_{i \in [N]}$ for each vertex $v \in V.$ Further,  $\tilde{n}=|V|=\nfrac{2^N}{N}$ and ${\rm opt}(\calU_\eta) \leq \log^{-\eta} \tilde{n}.$ 

\medskip
The reduction implies that $n=   |V^*| = 2^N \cdot  |V| \leq O(\tilde{n}^2 \log \tilde{n})$ and, hence, $\log^{-\eta} \tilde{n} \approx \log^{-\eta} n$ up to a constant.    
Thus, if $\log^{-\eta} n \leq 2^{-O(\nfrac{1}{\epsilon^2})},$ then it follows from Lemma \ref{lem:bs-s} and the discussion in Section \ref{sec:reduction}
that every cut in $\mathcal{I}_\epsilon$ that cuts at least a constant fraction of demand pairs cuts at least $\sqrt{\epsilon}$ fraction of edges. This proves the first claim. A constraint on $\eta,$ as we see shortly, is that $\eta \leq \epsilon.$ Thus, choosing $\eta=\epsilon$ implies that in order to ensure
$\log^{-\epsilon} {n} \leq 2^{-O(\nfrac{1}{\epsilon^2})},$ it is sufficient to set $\epsilon$ to be $(\log \log  {n})^{\nfrac{-1}{3}}.$

\medskip
Thus, to complete the proof of Theorem \ref{thm:bs-restate}, it remains to construct vectors $\V_u$ for each vertex $u \in V^*$ that satisfy the required constraints and have a small objective value. This is the focus of this section. Here again the starting point is the SDP solution to the {\UG} integrality gap $\calU_\eta.$ Recall that the vectors $\{v_i \}_{i \in [N]}$ form an orthonormal basis of $\mathbb{R}^N$ for each $v \in V$ and, in addition satisfy Triangle Inequality, the Matching Property and the Closeness Property in Theorem \ref{thm:ulc-sdp}. In addition, the SDP objective value of these vectors for $\calU_\eta$ is $1-9\eta.$

\medskip
For each vertex $v \in V$ there is a block of vertices $\{(v,x): x \in \{-1,1\}^N\}$ in $V^*.$ Thus, we need a unit vector for each $(v,x).$
A choice for such a vector is
\begin{equation}\label{eq:vec1}
\V_{(v,x)} \defeq \frac{1}{\sqrt{N}} \sum_{i \in [N]} x_i v_i^{\otimes 2}.
\end{equation}
The fact that this is a unit vector is easy to see. Recall that for a typical edge in $\calU_\eta,$ the basis vectors are $\eta$-close when matched according to the permutation corresponding to that edge. Further, recall that for an edge between $(v,x)$ and $(w,y),$ there must be an edge between $v$ and $w$ in $\calU_\eta.$ Moreover, for a typical edge in $\mathcal{I}_\epsilon,$ except with probability $\epsilon,$ the relative Hamming distance between $x$ and $y$ is at most $2\epsilon$ (after taking into account the permutation between $v$ and $w$ in $\calU_\eta$). This easily implies that for a typical edge in $\mathcal{I}_\epsilon,$
$$ \left\langle \V_{(v,x)}, \V_{(w,y)} \right \rangle \geq 1- O(\eta+\epsilon).$$
Since the vectors are of unit length, this implies that
$$ \left\| \V_{(v,x)}- \V_{(w,y)} \right\|^2 \leq O(\eta+\epsilon).$$
This is what dictates the choice of $\eta=\epsilon$ and we obtain that our SDP solution to $\mathcal{I}_\epsilon$ has an objective value at most $O(\epsilon).$ To see the well-separatedness of this SDP solution, observe that for each $v \in V$,  $\V_{(v,x)}$ and $\V_{(v,-x)}$ are unit vectors in opposite direction.

\medskip
It remain to prove that the vectors $\left\{\V_{(v,x)}\right\}$ satisfy the triangle inequality. This is the technically hardest part of the paper and is shown via an extensive case analysis that repeatedly uses the fact that the vectors for $\calU_\eta$ satisfy the properties they do.
In fact, we do not know whether  the vectors described above work for this proof. We need to modify the vectors in \eqref{eq:vec1} as follows
$$ \left( \frac{1}{\sqrt{N}}\sum_{i \in [N]} x_i v_i^{\otimes 8} \right)^{\otimes (2^{240}+1)}.$$
While the inner tensor, which goes to $8$ from $2$, is a minor modification, it ensures that when we take inner products
of the form
$$ \left \langle \frac{1}{\sqrt{N}} \sum_{i \in [N]} v_{i}^{\otimes 8},  \frac{1}{\sqrt{N}} \sum_{i' \in [N]} w_{i'}^{\otimes 8} \right\rangle,
$$
and if $\langle v_i,w_i \rangle \approx 1-\eta$ for all $i \in [N],$ then the contribution of the cross terms is negligible and the inner product remains around $1-\eta.$ This $8$-th tensor also implies the converse: If 
$$ \left \langle  \frac{1}{\sqrt{N}} \sum_{i \in [N]}  v_{i}^{\otimes 8},  \frac{1}{\sqrt{N}} \sum_{i' \in [N]} w_{i'}^{\otimes 8} \right\rangle \geq 1-\eta,$$ then there is a permutation $\pi:[N] \mapsto [N]$ such that for all $i \in [N],$
 $$|\langle v_{\pi(i)},w_{i} \rangle| \geq 1-2\eta.$$ This latter property and
the outer tensor are crucial in the proof of the triangle inequality.\footnote{This property has also been key in the results of Arora {et al.} \cite{AroraKKSTV08}.}
This new SDP solution is also easily seen to satisfy the properties satisfied by the previous SDP solution up to a loss of an additional constant factor.

\medskip
We conclude this overview by giving the reader some idea of why we have the outer tensor.  Start by noting that proving the triangle inequality is the same as showing
 $$ 1 + \langle \V_{u,x},\V_{v,y}\rangle^t \geq \langle \V_{u,x},\V_{w,z}\rangle^t +
        \langle \V_{v,y},\V_{w,z}\rangle^t$$
since all the vectors have unit length.
If none of  the dot-products has magnitude
at least $\nfrac{1}{3}$  the inequality  holds trivially. Thus, we may assume that one of the inner products, say,  $|\langle \V_{v,y},\V_{w,z}\rangle|^t \geq \nfrac{1}{3}$. This implies
that
 $|\langle  \V_{v,y},\V_{w,z} \rangle| = 1-O(\nfrac{1}{t})$.
By the converse property mentioned earlier, it can be deduced  that, for some $i_0,j_0 \in [N],$  $|\langle \v_{i_0},\w_{j_0}\rangle| =  1-O(\nfrac{1}{t})$ which can be made very close to $1$ by picking $t$ large enough. This turns out to be convenient towards proving the triangle inequality via a case analysis, see Lemma \ref{lemma:allclose1}. 

\medskip
Unfortunately,  we cannot provide much more intuition than this and,
as mentioned in the introduction, for a more intuitive proof of the triangle inequality one can refer to the papers \cite{KhotS09,RaghavendraS09}.
We now present the graph construction and the SDP solution formally and prove the claims above for the SDP solution.

\subsection{The Graph}\label{sec:INTGAP-graphs}
We recall the following
notations which are needed. For a permutation
$\pi:[N]\mapsto [N]$ and a vector $\x\in \{-1,1\}^N,$ the vector
$\x \circ \pi$ is defined to be the vector with its $j$-th entry
as $(\x \circ \pi)_j \defeq \x_{\pi(j)}.$ For $\epsilon>0,$ the notation
$\x \in _\epsilon \{-1,1\}^N$ means that the vector $\x$ is a
random $\{-1,1\}^N$ vector, with each of its bits independently set
to $-1$ with probability $\epsilon,$ and set to $1$ with probability
$1-\epsilon.$

\medskip
The {\BS} instance has a parameter $\epsilon>0$ and we refer to it as
${\cal I}_\epsilon(V^*,E^*).$ We start with the {\UG}
instance $\calU_\eta=\left(G(V, E), [N], \{ \pi_e\}_{e \in E}, {\rm
wt} \right)$ of Theorem \ref{thm:ulc-sdp}. In ${\cal
I}_\epsilon,$ each vertex $v \in V$ is replaced by a
{block} of vertices denoted by $V^*[v]$. This block consists
of vertices $(v,\x)$ for each $\x \in \{-1,1\}^N.$ Thus, the set of
vertices for the {\BS} instance is
$$ V^* \defeq \{(v,\x) \ | \ v \in V, \ \x \in \{-1,1\}^N \}
\quad \quad {\rm and} \quad \quad V^* = \cup_{v \in V} V^*[v].$$
The edges in the {\BS} instance are defined as follows: For
$e\{v,w\}\in E,$ there is an edge $e^*$ in ${\cal I}_\epsilon$ between $(v,\x)$ and $(w,\y),$ with weight
$$ {\rm wt}_{\rm BS}(e^*) \defeq
{\rm wt}(e) \cdot \Pr_{\stackrel{\x' \in_{\nfrac{1}{2}} \{-1,1\}^N} {\bmu
\in _\epsilon \{-1,1\}^N}} \left[\left(\x=\x'\right) \wedge \left(
\y=\x'{\bmu} \circ \pi_e\right) \right].$$ Notice that the size of
${\cal I}_\epsilon$ is $|V^*| = |V|\cdot 2^N =  O({\tilde{n}^2}\log \tilde{n}).$
The following theorem establishes that every cut in ${\cal
I}_\epsilon$ that cuts a constant fraction of the demand cuts a large fraction of the edges. It is a restatement of Lemma \ref{lem:bs-s}.
See Section \ref{sec:inner-bs} for details.

\begin{theorem}[No Small Balanced Cut]\label{thm:bs-instance-graph} For every $t \in
(\nfrac{1}{2}, 1)$, there exists a constant $c_t > 0$ such that the
following holds: Let $\epsilon > 0$ be sufficiently small and
let $\: \calU_\eta\left(G(V, E), [N], \{ \pi_e\}_{e \in E}, {\rm wt}
\right)$ be an instance of {\UG} with ${\rm opt}(\calU_\eta) <
2^{-O(\nfrac{1}{\epsilon^2})}$. Let ${\cal I}_\epsilon$ be the
corresponding instance of {\BS} as defined above. Let $V^* = \cup_{v \in V}
V^*[v]$ be the partition of its vertices as above. Then, any
$\nfrac{5}{6}$-piecewise balanced cut $(A,\overline{A})$ in ${\cal
I}_\epsilon$ (in the sense of Definition
\ref{def:part-bl}) satisfies
$$ \sum _{e^* \in E^*(A,\overline{A})} {\rm wt}_{\rm BS}(e^*) \geq c_t\epsilon^{t}.
$$
\end{theorem}

\subsection{The {SDP} Solution}\label{sec:INTGAP-sdp}
Now we present an {SDP} solution for ${\cal I}_{\epsilon}(V^*,E^*, {\rm wt}_{\rm BS})$ that satisfies Properties (3),
(4) and (5) of Theorem \ref{thm:graph-construction}. This 
proves Theorem \ref{thm:graph-construction} and, hence, Theorem \ref{thm:bs-restate}.

\medskip
We begin with the {SDP} solution of Theorem \ref{thm:ulc-sdp}.
Recall that $[N]$ is identified with the group $\mathbb{F}_2^k$ where $N=2^k,$ and $\oplus$ is the corresponding group operation.
We construct the following unit vectors, one for each pair
$(v,\x),$ where $v \in V$ and $\x\in \{-1,1\}^N$ (note that $V$
is the set of vertices of the {\UG} instance of Theorem
\ref{thm:ulc-sdp}):
\begin{equation}\label{eq:vectors}
\V_{v,\x} \defeq \frac{1}{\sqrt{N}} \sum _{i \in [N]} \x_i\v_i ^{\otimes 8}.
\end{equation}
For $(v,\x)\in V^*,$ we associate the vector
$\V_{v,\x}^{\otimes t},$ where $t=2^{240}+1$.
We start by noting that this vector is indeed a unit vector.
Since $\{v_i  \}_{i \in [N]}$ is an orthonormal basis for $\mathbb{R}^N$ and $x_i \in \{-1,1\},$
$$\left\langle \V_{v,\x} , \V_{v,\x} \right\rangle = \frac{1}{N} \sum_{i \in [N]} x_i^2 \langle v_i,v_i \rangle ^8 = \frac{1}{N} \sum_{i \in [N]} 1 = 1.$$
Hence, for every $v \in V$ and $\x \in \PM,$
\begin{equation}\label{eq:norm1}
\|\V_{v,\x}^{\otimes t} \|=1.
\end{equation}

\noindent
Next, we show  Property (5) in Theorem \ref{thm:graph-construction} which establishes that the SDP solution has value $O(\epsilon)$ when $\eta=\epsilon.$
\begin{theorem}[Low Objective Value]\label{thm:bs-complete}
$\sum _{e^*\{(v,\x),(w,\y)\} \in E^*} {\rm wt}_{\rm BS}(e^*)
\| \V_{v,\x}^{\otimes t}-\V_{w,\y}^{\otimes t} \|^2 \leq O(\eta+\epsilon).$
\end{theorem}

\noindent
The proof of this theorem uses the following lemma which shows that, if $e\{v,w\}$ is an edge in the {\UG} instance $\calU_\eta$ so that the corresponding orthonormal bases are $\eta$-close (via the permutation $\pi$), then $\V_{v,x}$ and $\V_{w,y}$ are also close if $x \circ \pi$ and $y$ are close.

\begin{lemma}\label{lem:innerprod} Let $0<\eta<\nfrac{1}{2}$ and
assume that for $v,w \in V$ and $i_0,j_0\in [N],$ $\langle
\v_{i_0},\w_{j_0} \rangle=1-\eta.$ Let $\pi : [N] \mapsto [N]$
be defined to be $\pi(j_0 \oplus j ) \defeq i_0\oplus j \ \forall \ j \in
[N]$.
Then,
\begin{itemize}
\item {\bf Lower Bound:}
$(1-\eta)^{8}(1-2{\Delta}(\x\circ\pi,\y))-(2\eta)^{4} \leq \langle \V_{v,\x},\V_{w,\y} \rangle.$
\item {\bf Upper Bound:} $\langle \V_{v,\x},\V_{w,\y} \rangle
\leq (1-\eta)^{8}(1-2{\Delta}(\x\circ\pi,\y))+(2\eta)^{4}.$
\end{itemize}
Here, ${\Delta}(\x,\y)$ denotes the fraction of points where $\x$ and $\y$ differ.
\end{lemma}

\noindent
We first show how Lemma \ref{lem:innerprod} implies Theorem \ref{thm:bs-complete}.

\begin{proof}[of Theorem \ref{thm:bs-complete}]
It is sufficient to prove that for an edge $e\{v,w\} \in E$ picked
with probability ${\rm wt}(e)$ (from the {\UG} instance
$\calU_\eta$), $\x \in _{\nfrac{1}{2}} \{-1,1\}^N,$ and ${\mu}
\in_\epsilon \{-1,1\}^N,$
$$\E_{e\{v,w\}}\left[ \E_{\begin{subarray}{c} \x \in _{\nfrac{1}{2}} \{-1,1\}^N \\
{ {\mu}} \in _\epsilon \{-1,1\}^N \end{subarray}} \left\langle
\V_{v,\x}^{\otimes t}, \V_{w,\x{{\mu}}\circ \pi_e }^{\otimes t}\right\rangle \right]
\geq 1-O(t(\eta+\epsilon)).$$
Since $e\{v,w\}$ is an edge of $\calU_\eta,$ we know
from the Closeness Property of Theorem \ref{thm:ulc-sdp}, that  there are $i_0,j_0 \in [N]$
such that $\langle \v_{i_0}, \w_{j_0} \rangle \geq 1-O(\eta).$
Moreover, $\pi_e(j_0\oplus j)=i_0\oplus j, \ \forall \ j \in [N]$.
Further, it follows from a simple Chernoff Bound argument that, except with
probability $\epsilon$, ${\Delta}(\x,\x{\mu}) \leq
2\epsilon.$ Thus, using the lower bound estimate from Lemma
\ref{lem:innerprod},
we get that
$$ \left\langle \V_{v,\x}^{\otimes t}, \V_{w,\x{\mu}\circ\pi_e}^{\otimes t}\right\rangle \geq
1-O(t(\eta+\epsilon)).$$ This completes the proof.
\end{proof}

\noindent
We now present the proof of Lemma \ref{lem:innerprod}.

\begin{proof}[of Lemma \ref{lem:innerprod}]
Note that
\begin{eqnarray*} \langle \V_{v,\x},\V_{w,\y} \rangle &=& \frac{1}{N} \sum_{i,i' \in [N]} \x_i\y_{i'}\langle \v_i,\w_{i'} \rangle ^{8}\\ &=&\frac{1}{N} \sum_{i,i' \in [N]} \x_{i_0\oplus i}\y_{j_0\oplus i'}\langle \v_{i_0\oplus i},\w_{j_0\oplus i'} \rangle ^{8}.
\end{eqnarray*}
We first show that in the above summation, terms with $i=i'$ dominate the summation.
\smallskip
\noindent
Since $\langle \v_{i_0},\w_{j_0} \rangle=1-\eta,$ the Matching Property implies that for all $i \in [N],$ $\langle \v_{i_0 \oplus i},\w_{j_0 \oplus i} \rangle= 1-\eta.$
Further, since the vectors $\{ \w_{i'} \}_{i'\in [N]}$ form an orthonormal
basis for $\mathbb{R}^{N},$
$ \sum_{i' \in [N]} \ \langle \v_{i_0\oplus i},\w_{j_0 \oplus i'}\rangle^2 = 1.$ Hence,
$$
\sum_{i' \in [N], i'\neq i}
\ \langle\v_{i_0 \oplus i}, \w_{j_0\oplus i'}\rangle^{8} \leq \left(1 - (1- \eta)^2\right)^{4} =
(2 \eta - \eta^2)^4 \leq (2\eta)^4.$$
Now, $\langle \V_{v,\x},\V_{w,\y} \rangle$ is at least
$$
\frac{1}{N} \sum_{i \in [N]} \x_{i_0\oplus i} \y_{j_0 \oplus i} \langle \v_{i_0\oplus i},\w_{j_0 \oplus i} \rangle ^{8}
- \frac{1}{N} \sum_{\begin{subarray}{c} i, i' \in [N] \\ i\not= i' \end{subarray}}
\langle \v_{i_0\oplus i},\w_{j_0 \oplus i'} \rangle ^{8},$$ and at most
$$
\frac{1}{N} \sum_{i \in [N]} \x_{i_0\oplus i} \y_{j_0 \oplus i} \langle \v_{i_0\oplus i},\w_{j_0\oplus i} \rangle ^{8}
+ \frac{1}{N} \sum_{\begin{subarray}{c} i, i' \in [N] \\ i\not= i' \end{subarray}}
\langle \v_{i_0\oplus i},\w_{j_0\oplus i'} \rangle ^{8}.$$
The first term in both these expressions is
$$ \frac{1}{N} \sum_{i \in [N]} \x_{i_0\oplus i} \y_{j_0\oplus i} (1-\eta)^{8}= (1-2{\Delta}(\x\circ \pi,\y)) (1-\eta)^{8}.$$
The second term is bounded by $(2\eta)^4$ as seen above.
This completes the proof of the lemma.
\end{proof}

\noindent
The well-separatedness of the SDP solution, or Property (4) in Theorem \ref{thm:graph-construction}, follows
from the following lemma.
\begin{lemma}[Well Separatedness] \label{lem:well-separated}
For any odd integer $t>0,$
$$
\frac{1}{2} \E_{\x \in_{\nfrac{1}{2}} \PM, \ \y \in_{\nfrac{1}{2}} \PM}\left[\|
\V_{v,\x}^{\otimes t}- \V_{v,\y}^{\otimes t} \|^2 \right] =1.$$
\end{lemma}

\begin{proof}
Observe that
\begin{eqnarray*}
\frac{1}{2} \E_{\x,\y} \left[\| \V_{v,\x}^{\otimes t}- \V_{v,\y}^{\otimes t} \|^2 \right]&=&
\E_{\x,\y} \left[ 1- \langle \V_{v,\x}^{\otimes t}, \V_{v,\y}^{\otimes t} \rangle \right] \\
& = & 1- \E_{\x,\y}\left[ \left( \frac{1}{N} \sum_{i, j \in [N]} \x_i \y_j\langle \v_{i}, \v_{j}\rangle^{8} \right)^t \right] \\
& = & 1.
\end{eqnarray*}
The last equality follows from the fact that the contribution of
$(\x,\y)$ to the expectation is canceled by that of $(\x,-\y).$
\end{proof}

\noindent
Finally, the following theorem establishes that our SDP solution satisfies the triangle inequality, Property (3)
of Theorem \ref{thm:graph-construction}.

\begin{theorem}[Triangle Inequality]\label{thm:metric}
For $t=2^{240}+1,$ the set of vectors $\{\V_{v,\x}^{\otimes t}\}_{v \in V, \x \in \PM}$ give rise to a negative-type metric.
\end{theorem}

\paragraph{Proof of Theorems \ref{thm:graph-construction} and \ref{thm:bs-restate}.}
Before we go into the proof of Theorem \ref{thm:metric}, we note that  Theorem
\ref{thm:bs-instance-graph} and Theorem \ref{thm:bs-complete}, along with \eqref{eq:norm1}, Lemma \ref{lem:well-separated} and Theorem \ref{thm:metric}, for the choices $\epsilon= (\log
\log n)^{\nfrac{-1}{3}}$ and $\eta = O(\epsilon)$ complete the proof of Theorem \ref{thm:graph-construction} and Theorem \ref{thm:bs-restate} 
(note that ${\rm opt}(\calU_\eta) \leq \log^{-\eta} \tilde{n} \leq 3\log^{-\eta} n$).

\medskip

\paragraph{Proof of Theorem \ref{thm:metric}.} Theorem \ref{thm:metric} requires proving  that any three vectors $\V_{u,\x}^{\otimes t}$, $\V_{v,\y}^{\otimes t}$ and
$\V_{w, \z}^{\otimes t}$ satisfy
\begin{equation}\label{eq:tr1}
 1 + \langle \V_{u,\x}^{\otimes t},\V_{v,y}^{\otimes t} \rangle  \geq \langle \V_{u,\x}^{\otimes t},\V_{w,\z}^{\otimes t} \rangle
  +\langle \V_{v,\y}^{\otimes t},\V_{w,\z}^{\otimes t}\rangle.
\end{equation}
We can assume that at least one of the dot-products has magnitude
at least $\nfrac{1}{3}$; otherwise, the inequality  holds trivially. Assume,
w.l.o.g., that $$|\langle \V_{v,\y}^{\otimes t},\V_{w,\z}^{\otimes t}\rangle| \geq \nfrac{1}{3}.$$ This implies
that $|\langle \V_{v,y},\V_{w,z}\rangle|^t \geq \nfrac{1}{3},$ and therefore,
 $$|\langle  \V_{v,y},\V_{w,z} \rangle| = 1-\eta',$$  for some $\eta' = O(\nfrac{1}{t})$.
It follows that, for some $i,j \in [N],$  $|\langle \v_i,\w_j\rangle| =  1-\beta$ for some
$\beta  \leq 2^{-160}.$
We give a quick proof of this.
Let $i_0,j_0$ be $\arg\max_{i,j} |\langle v_i, w_j\rangle| $ and $1-\beta = |\langle v_{i_0}, w_{j_0}\rangle|.$
Then,
$$
 \left| \langle  \V_{v,y},\V_{w,z} \rangle \right| = \left| \frac{1}{N} \sum_{i,j \in [N]} y_iz_j \langle v_i, w_j\rangle ^8 \right|\leq  \frac{1}{N} \sum_{i\in [N]} \langle v_{i_0\oplus i }, w_{j_0 \oplus i}\rangle ^8 +  \frac{1}{N} \sum_{i \neq j \in [N]}\langle v_{i_0\oplus i}, w_{j_0 \oplus j}\rangle ^8.$$
By the Matching Property, $\langle v_{i_0\oplus i }, w_{j_0 \oplus i}\rangle= \langle v_{i_0}, w_{j_0}\rangle$ for all $i \in [N].$ Hence, $$\frac{1}{N} \sum_{i \in[N]} \langle v_{i_0\oplus i}, w_{j_0 \oplus i}\rangle ^8  = (1-\beta)^8.$$
 Moreover, by orthonormality, for all $i \in [N],$ 
$$
\sum_{i' \in [N], i'\neq i}
\ \langle\v_{i_0 \oplus i}, \w_{j_0\oplus i'}\rangle^{8} \leq \left(1 - (1- \beta)^2\right)^{4} =
(2 \beta - \beta^2)^4 \leq (2\beta)^4.$$
Thus, 
$$1-\eta' =  \left| \langle  \V_{v,y},\V_{w,z} \rangle \right| \leq (1-\beta)^8 + (2 \beta)^4,$$ giving us the claimed upper bound on $\beta.$
 By relabeling, if necessary, we may assume that
$|\langle \v_{1},\w_{1} \rangle|=1-\beta.$

\medskip
\noindent
Note that \eqref{eq:tr1} is equivalent to showing  that
 $$ 1 + \langle \V_{u,x},\V_{v,y}\rangle^t \geq \langle \V_{u,x},\V_{w,z}\rangle^t +
        \langle \V_{v,y},\V_{w,z}\rangle^t. $$
The following elementary lemma, whose proof appears at the end of this section, implies that it is sufficient to prove that
\begin{equation}\label{eq:tr2}
  1 + \langle \V_{u,x},\V_{v,y}\rangle\geq \langle \V_{u,x},\V_{w,z}\rangle +
        \langle \V_{v,y},\V_{w,z}\rangle.
\end{equation}

\begin{lemma} \label{lemma:pow-t1}
Let $a,b,c \in [-1,1]$ such that $1+a \geq b +c$.
Then, $1+a^t \geq b^t+c^t$ for every odd integer $t \geq 1$.
\end{lemma}

\noindent
Equation \eqref{eq:tr2} is the same as showing
$$ N + \sum_{i,j=1}^N x_i y_j \langle \u_i,\v_j\rangle^{8} \geq \sum_{i,j=1}^N
    x_i z_j \langle \u_i,\w_j\rangle^{8} + \sum_{i,j=1}^N y_i z_j \langle \v_i,\w_j\rangle^{8}.$$
As noted before, we may assume that  $|\langle \v_1,\w_1 \rangle|=1-\beta$ and, hence,
by the Matching Property,
 $$ \langle \v_1,\w_1 \rangle= \langle \v_2,\w_2\rangle = \cdots=\langle \v_N,\w_N\rangle = \pm (1-\beta).$$
Let $\lambda \defeq \max_{1\leq i, j  \leq N} |\langle \u_i,\w_j\rangle|.$ We may
assume, w.l.o.g., that the maximum is achieved for $\u_1, \w_1,$ and again by the Matching Property,
$$ \langle \u_1,\w_1\rangle = \langle \u_2,\w_2 \rangle = \cdots =\langle \u_N,\w_N \rangle= \pm\lambda.
  $$
Now, Theorem \ref{thm:metric} follows from the following  lemma.

\begin{lemma} \label{lemma:allclose1}
 Let $\{\u_i\}_{i=1}^N, \{\v_i \}_{i=1}^N,
\{\w_i\}_{i=1}^N$ be three sets of unit vectors in $\Real^N,$
such that the vectors
in each set are mutually orthogonal. Assume that any three of
these vectors satisfy the triangle inequality. Assume, moreover, that
\begin{gather*}
  \langle \u_1, \v_1 \rangle= \langle \u_2, \v_2 \rangle= \cdots = \langle \u_N , \v_N \rangle, \\
 \; \;   \lambda \defeq \langle \u_1,\w_1 \rangle =  \langle \u_2, \w_2 \rangle= \cdots = \langle \u_N,\w_N \rangle
   \geq 0, \\
 \forall 1\leq i,j \leq N, \; \;  |\langle \u_i, \w_j\rangle | \leq  \lambda, \\
 1-\beta \defeq \langle \v_1, \w_1 \rangle=  \langle \v_2, \w_2 \rangle=  \cdots = \langle \v_N , \w_N \rangle
 ,
\end{gather*}
 where $0 \leq \beta \leq 2^{-160}.$  Let
 $x_i, y_i, z_i \in \{-1,1\}$ for $1\leq i \leq N$. Define unit vectors
  $$ \u \defeq  \frac{1}{\sqrt{N}} \sum_{i=1}^N x_i \u_i^{\otimes 8},
  \quad \v \defeq  \frac{1}{\sqrt{N}}  \sum_{i=1}^N y_i \v_i^{\otimes 8} \quad
  \w \defeq \frac{1}{\sqrt{N}}  \sum_{i=1}^N z_i \w_i^{\otimes 8} .$$
Then, the vectors $\u,\v,\w$ satisfy the triangle inequality $1+\langle \u,\v \rangle \geq
\langle \u,\w \rangle + \langle \v,\w\rangle,$ i.e.,
\begin{equation*}   N + \sum_{i,j=1}^N x_i y_j \langle \u_i, \v_j\rangle^{8} \geq \sum_{i,j=1}^N
    x_i z_j \langle \u_i, \w_j\rangle^{8} + \sum_{i,j=1}^N y_i z_j \langle \v_i, \w_j\rangle^{8}.\end{equation*}
\end{lemma}

\noindent
Note that we only have $|\langle v_1,w_1\rangle|=1-\beta$ but we can remove the absolute value and use this lemma as it holds for {\em all} sign patterns $x_i,y_i,z_i.$ The  proof of this  lemma is very technical and  appears in Appendix \ref{sec:sdp-triangle}.
We conclude with a proof of Lemma \ref{lemma:pow-t1}.

\begin{proof}[of Lemma \ref{lemma:pow-t1}]
First, we notice that it is sufficient to prove this inequality when $0 \leq a,b,c \leq 1.$
Suppose that $b <0$ and $c < 0,$ then $b^t+c^t < 0 \leq 1+a^t.$ Hence, without loss of generality assume that $b \geq 0.$ If $c<0$ and $a \geq 0,$ then $b^t+c^t < b^t \leq 1+ a^t.$ If $c <0$ and $a <0,$ by hypothesis, $1-c \geq b-a,$ which is the same as $1+|c| \geq b+|a|,$ and proving $1+a^t \geq b^t+c^t$ is equivalent to proving $1+|c|^t \geq b^t+|a|^t.$  Hence, we may assume that $c \geq 0.$ If $a <0,$ then $1+a^t= 1-|a|^t \geq 1-|a|=1+a \geq b+c \geq b^t+c^t.$  Hence, we may assume that $0 \leq a,b,c \leq 1.$

Further, we may assume that $a<b \leq c.$ Since, if $a\geq b,$ then $1+a^t \geq c^t+b^t.$ $1+a \geq b+c$ implies that $1-c \geq b-a.$ Notice that both sides of this inequality are positive. It follows from the fact that $0\leq a<b \leq c \leq 1,$ that
$\sum_{i=0}^{t-1} c^i \geq \sum_{i=0}^{t-1}a^{i}b^{t-1-i}.$ Multiplying these two inequalities, we obtain $1-c^t \geq b^t-a^t,$ which implies that $1+a^t \geq b^t+c^t.$ This completes the proof.

\end{proof}

\paragraph{\bf Acknowledgments.}

We would like  to thank Assaf Naor and James Lee for ruling out some of our initial approaches. Many
thanks to Sanjeev Arora, Moses Charikar, Umesh Vazirani, Ryan
O'Donnell and Elchanan Mossel for insightful discussions at
various junctures.

\bibliographystyle{plain}
\bibliography{ug-sc}

\appendix

\section{Proof of Soundness of the PCP Reduction}\label{sec:app-pcp}

\begin{lemma}[Same as Lemma \ref{lem:bs-s}] For every $t \in
(\nfrac{1}{2}, 1)$, there exists a constant $b_t > 0$ such that the
following holds: Let $\epsilon > 0$ be sufficiently small and let
$\: \calU$  be an instance of {\UG} with ${\rm opt}(\calU) <
2^{-O(\nfrac{1}{\epsilon^2})}.$ Then, for every $\nfrac{5}{6}$-piecewise balanced
proof $\Pi,$  $$\; \Pr\left[ V_\epsilon\;\rm{accepts} \;
\Pi \right] < 1-b_t\epsilon^t.$$
\end{lemma}

\begin{proof}
The proof is by contradiction: We assume that there is a $\nfrac{5}{6}$-piecewise balanced proof $\Pi,$ which the verifier accepts with probability at least $1-b_t\epsilon^t,$ and deduce that
${\rm opt}(\calU) \geq
2^{-O(\nfrac{1}{\epsilon^2})}.$  We let $b_t\defeq \frac{1-e^{-2}}{96}c_t,$ where $c_t$ is the constant in Bourgain's Junta theorem.

\noindent
The probability of acceptance of the verifier is
$$  \frac{1}{2} +\frac{1}{2}\E_{v,e\{v,w\} , \x, \bmu }\left[
 A^v(\x)A^w(\x\bmu \circ \pi_e)\right]. $$
Using the Fourier expansion  $A^v = \sum_\alpha \widehat{A}^v_\alpha
\chi_\alpha$ and $A^w= \sum_\beta \widehat{A}^w_\beta \chi_\beta,$ and the orthonormality of characters, we get that this probability is
$$  \frac{1}{2} +\frac{1}{2}\E_{v,e\{v,w\}}\left[
 \sum _\alpha \widehat{A}^v_\alpha \widehat{A}^w_{\pi_e^{-1}(\alpha)} (1-2\epsilon)^{|\alpha|}\right]. $$
Here $\alpha \subseteq [N].$ Hence, the acceptance probability is
$$  \frac{1}{2} +\frac{1}{2}\E_{v}\left[
 \sum _\alpha \widehat{A}^v_\alpha \E_{e\{v,w\}} \left[\widehat{A}^w_{\pi_e^{-1}(\alpha)} \right] (1-2\epsilon)^{|\alpha|}\right]. $$
  If this
acceptance probability is at least $1-b_t\epsilon^t,$   then,
$$\E_{v}\left[
 \sum _\alpha \widehat{A}^v_\alpha \E_{e\{v,w\}} \left[\widehat{A}^w_{\pi_e^{-1}(\alpha)} \right] (1-2\epsilon)^{|\alpha|}\right] \geq 1-2b_t\epsilon^t.$$
Hence,
over the  choice of $v$, with probability at least $\nfrac{23}{24},$
$$\sum _\alpha \widehat{A}^v_\alpha \E_{e\{v,w\}} \left[\widehat{A}^w_{\pi_e^{-1}(\alpha)} \right] (1-2\epsilon)^{|\alpha|} \geq  1-48b_t\epsilon^t.$$
Call such vertices $v\in V$ {\em good}.
 Fix  a good vertex $v.$
Using the Cauchy-Schwarz inequality we get,
$$ \sum _\alpha \widehat{A}^v_\alpha \E_{e\{v,w\}} \left[\widehat{A}^w_{\pi_e^{-1}(\alpha)} \right] (1-2\epsilon)^{|\alpha|} \leq \sqrt{\sum_\alpha \left(\widehat{A}^v_\alpha\right)^2 (1-2\epsilon)^{2 |\alpha|}\sum_\alpha \E_{e\{v,w\}}^2 \left[\widehat{A}^w_{\pi_e^{-1}(\alpha)}\right]}.$$ Combining Jensen's inequality and Parseval's identity, we get that  $$\sum_\alpha \E_{e\{v,w\}}^2 \left[\widehat{A}^w_{\pi_e^{-1}(\alpha)}\right] \leq 1.$$ Hence,
$$
 1-96b_t\epsilon^t \leq \sum_\alpha \left(\widehat{A}^v_\alpha\right)^2 (1-2\epsilon)^{2 |\alpha|}.
$$
Now we combine Parseval's identity with the fact that $1-x \leq e^{-x}$ to obtain
$$\sum_{\alpha \ : \  |\alpha| > \nfrac{1}{\epsilon} }\left(\widehat{A}^v_\alpha\right)^2 \leq \frac{96}{1-e^{-2}}b_t\epsilon^t = c_t\epsilon^t.$$
Hence, by Bourgain's Junta theorem
$$
\sum_{\alpha\ : \ |\widehat{A}^v_\alpha| \leq
\frac{1}{50}4^{\nfrac{-1}{\epsilon^2}}}\left(\widehat{A}^v_\alpha\right)^2 \leq
\frac{1}{2500}.
$$
Call $\alpha$ {\em good} if $\alpha \subseteq [N]$  is nonempty,
$|\alpha| \leq \epsilon^{-1}$ and $|\widehat{A}^v_\alpha| \geq
\frac{1}{50} 4^{\nfrac{-1}{\epsilon^2}}.$

\medskip
\noindent
{\bf Bounding the contribution due to large sets.}
\smallskip

\noindent Using the Cauchy-Schwarz inequality, Parseval's identity
and Jensen's inequality, we get
$$
 \left| \sum_{\alpha \ : \ |\alpha| >\nfrac{1}{\epsilon}}
 \widehat{A}^v_\alpha \E_{e\{v,w\}} \left[ \widehat{A}^w_{\pi_e^{-1}(\alpha)} \right](1-2\epsilon)^{|\alpha|} \right|
\leq
 \sqrt{ \sum_{\alpha \ : \ |\alpha| > \nfrac{1}{\epsilon}}
\left(\widehat{A}^v_\alpha\right)^2}  < \sqrt{c_t\epsilon^t}.
$$
We can choose $\epsilon$ to be small enough so that the last term above is less than $\nfrac{1}{50}.$

\medskip\medskip
\noindent {\bf Bounding the contribution due to small Fourier
coefficients.}
\smallskip

\noindent Similarly, we use $\sum_{\alpha \ :\
|\widehat{A}^v_\alpha| \leq
\frac{1}{50}4^{\nfrac{-1}{\epsilon^2}}}\left(\widehat{A}^v_\alpha \right)^2 \leq
\nfrac{1}{2500},$ and  get
$$ \left| \sum_{\alpha \ : \  |\widehat{A}^v_\alpha| \leq
\frac{1}{50}4^{\nfrac{-1}{\epsilon^2}}}
 \widehat{A}^v_\alpha \E_{e\{v,w\}} \left[ \widehat{A}^w_{\pi_e^{-1}(\alpha)} \right](1-2\epsilon)^{|\alpha|} \right|
\leq   \frac{1}{50}.$$

\medskip\medskip
\noindent
{\bf Bounding the contribution due to the empty set.}
\smallskip

\noindent Since $\E_v \left[ |\widehat{A}^v _\emptyset |\right]\leq
\nfrac{5}{6},$ $\E_{v}\left[\E_{e \{v,w\}}\left[|\widehat{A}^v
_\emptyset \widehat{A}^w _\emptyset  | \right]\right]\leq
\nfrac{5}{6}.$ This is because each $|\widehat{A}^v _\emptyset| \leq
1.$ Hence, with probability at least $\nfrac{1}{12}$ over the
choice of $v,$ $\E_{e\{v,w\}}\left[|\widehat{A}^v _\emptyset
\widehat{A}^w _\emptyset |\right] \leq \nfrac{10}{11}.$ Hence, with
probability at least $\nfrac{1}{24}$ over the choice of $v,$ $v$ is
good and  $\E_{e\{v,w\}}\left[|\widehat{A}^v_\emptyset
\widehat{A}^w _\emptyset |\right] \leq \nfrac{10}{11}.$ Call such a
vertex {\em very good}.

\medskip\medskip
\noindent
{\bf Lower bound for a very good vertex with good sets.}
\nopagebreak

\noindent Hence, for a very good $v,$
\begin{equation}
\label{eqn:goodest}
\sum_{\alpha \; \text{is good}}
 \widehat{A}^v_\alpha \E_{e\{v,w\}} \left[ \widehat{A}^w_{\pi_e^{-1}(\alpha)} \right](1-2\epsilon)^{|\alpha|} \geq 1-\frac{1}{50}-\frac{1}{50}-\frac{10}{11}\geq \frac{1}{22}.
\end{equation}

\medskip\medskip
\noindent {\bf The labeling.}

\noindent
Now we define a labeling for the {\UG} instance $\calU$   as follows: For a  vertex $v \in V$, pick $\alpha$ with
probability $\left(\widehat{A}^v_\alpha\right)^2,$ pick a random element of $\alpha$
and define it to be the label of $v.$

\medskip
\noindent
Let $v$  be a very good vertex.
It follows
that the weight of the edges adjacent to $v$ satisfied by this labeling  is at least
$$
\E_{e\{v,w\}} \left[ \sum_{\alpha \; \text{is good}}
\left(\widehat{A}^v_\alpha\right)^2 \left(\widehat{A}^w_{\pi_e^{-1}(\alpha)}\right)^2
\frac{1}{|\alpha|} \right] \geq \epsilon \ \E_{e\{v,w\}}
\left[ \sum_{\alpha \; \text{is good}} \left(\widehat{A}^v_\alpha\right)^2
\left(\widehat{A}^w_{\pi_e^{-1}(\alpha)}\right)^2 \right].$$ This is at least
$$ \epsilon \ \frac{1}{2500}4^{\nfrac{-2}{\epsilon^2}}\; \E_{e\{v,w\}}
\left[ \sum_{\alpha \; \text{is good}} \left(\widehat{A}^w_{\pi_e^{-1}(\alpha)}\right)^2 \right],$$
which is at least
 $$
\epsilon \ \frac{1}{2500}4^{\nfrac{-2}{\epsilon^2}}\; \E_{e\{v,w\}}
\left[ \sum_{\alpha \; \text{is good}}
\left(\widehat{A}^w_{\pi_e^{-1}(\alpha)}\right)^2(1-2\epsilon)^{|\alpha|}
\right].$$ It follows from the Cauchy-Schwarz inequality and
Parseval's identity that this is at least
$$\epsilon \ \frac{1}{2500} 4^{\nfrac{-2}{\epsilon^2}}\; \E_{e \{v,w\}} \left[ \left| \sum_{\alpha \; \text{is good}} \widehat{A}^v_\alpha \widehat{A}^w_{\pi_e^{-1}(\alpha)}(1-2\epsilon)^{|\alpha|} \right|^2 \right].$$ Using Jensen's inequality, we get that this is at least
$$\epsilon \ \frac{1}{2500} 4^{\nfrac{-2}{\epsilon^2}}\; \left( \E_{e \{v,w\}}  \left[ \sum_{\alpha \;
\text{is good}} \widehat{A}^v_\alpha
\widehat{A}^w_{\pi_e^{-1}(\alpha)}(1-2\epsilon)^{|\alpha|} \right]
\right)^2 \geq \epsilon \ \frac{1}{2500}4^{\nfrac{-2}{\epsilon^2}}
\frac{1}{484}.$$ Here, the last inequality follows from our
estimate in Equation \eqref{eqn:goodest}. Since, with probability
at least $\nfrac{1}{24}$ over the choice of $v,$ $v$ is very good,
our labeling  satisfies edges with total weight at least
   $\Omega\left(\epsilon \ 4^{\nfrac{-2}{\epsilon^2}} \right).$   This completes the proof of the lemma.
\end{proof}

\section{Proof of Lemma \ref{lemma:allclose1} (Triangle Inequality Constraint)}\label{sec:sdp-triangle}

\begin{lemma} \label{lemma:allclose}[Same as Lemma \ref{lemma:allclose1} up to a renaming of variables]
 Let $\{\u_i\}_{i=1}^N, \{\v_i \}_{i=1}^N,
\{\w_i\}_{i=1}^N$ be three sets of unit vectors in $\Real^N,$
such that the vectors
in each set are mutually orthogonal. Assume that any three of
these vectors satisfy the triangle inequality. Assume, moreover, that
\begin{gather}
  \langle \u_1, \v_1 \rangle= \langle \u_2, \v_2 \rangle= \cdots = \langle \u_N , \v_N \rangle, \\
 \; \;   \lambda \defeq \langle \u_1,\w_1 \rangle =  \langle \u_2, \w_2 \rangle= \cdots = \langle \u_N,\w_N \rangle
   \geq 0, \\
 \forall 1\leq i,j \leq N, \; \;  |\langle \u_i, \w_j\rangle | \leq  \lambda, \\
 1-\eta \defeq \langle \v_1, \w_1 \rangle=  \langle \v_2, \w_2 \rangle=  \cdots = \langle \v_N , \w_N \rangle
 ,
\end{gather}
 where $0 \leq \eta \leq 2^{-160}.$  Let
 $s_i, t_i, r_i \in \{-1,1\}$ for $1\leq i \leq N$. Define unit vectors
  $$ \u \defeq  \frac{1}{\sqrt{N}} \sum_{i=1}^N s_i \u_i^{\otimes 8},
  \quad \v \defeq  \frac{1}{\sqrt{N}}  \sum_{i=1}^N t_i \v_i^{\otimes 8} \quad
  \w \defeq \frac{1}{\sqrt{N}}  \sum_{i=1}^N r_i \w_i^{\otimes 8} .$$
Then, the vectors $\u,\v,\w$ satisfy the triangle inequality $1+\langle \u,\v \rangle \geq
\langle \u,\w \rangle + \langle \v,\w\rangle,$ i.e.,
\begin{equation} \label{allclose-1}
  N + \sum_{i,j=1}^N s_i t_j \langle \u_i, \v_j\rangle^{8} \geq \sum_{i,j=1}^N
    s_i r_j \langle \u_i, \w_j\rangle^{8} + \sum_{i,j=1}^N t_i r_j \langle \v_i, \w_j\rangle^{8}.\end{equation}
\end{lemma}

\begin{proof} It suffices to show that for every $1\leq j \leq N$,
\begin{equation} \label{allclose-1.5}
  1 + \sum_{i=1}^N s_i t_j  \langle \u_i, \v_j\rangle^{8}  \geq \sum_{i=1}^N
    s_i r_j \langle \u_i, \w_j\rangle^{8} +  t_j r_j \langle \v_j, \w_j\rangle^{8} + \sum_{1\leq i \leq N,
    i \not= j} \langle \v_i, \w_j\rangle^{8}.
\end{equation}
We consider four cases depending on value of $\lambda$.

\medskip
\noindent {\bf (Case 1) $\lambda \leq \eta$ : }
Since $\langle \v_j, \w_j \rangle= 1-\eta$,
and $\sum_{1\leq i \leq N} \langle \v_i, \w_j\rangle ^2 = 1,$ we have
$$\sum_{1\leq i \leq N; i \not= j}  \langle \v_i, \w_j\rangle^{8} \leq (2\eta-\eta^2)^4.$$ Also,
$\sum_{i=1}^N \langle \u_i,\w_j\rangle ^{8} \leq \lambda^{6} \leq \eta^{6}$.
Moreover, for any $1 \leq i \leq N$, by the triangle inequality,
$$ 1 \pm \langle \u_i,\v_j \rangle \geq  \langle \v_j,\w_j\rangle  \pm \langle \u_i,\w_j\rangle \geq 1-\eta-\lambda
\geq 1-2\eta,$$
and therefore, $$|\langle \u_i,\v_j\rangle|\leq 2\eta.$$
Therefore,
 $\sum_{i=1}^N \langle \u_i,\v_j\rangle^{8} \leq (2\eta)^{6}$. Thus, it suffices to prove that
 $$ 1  \geq (2\eta)^{6} + \eta^{6} + (1-\eta)^{8} + (2\eta-\eta^2)^4. $$
This is true when  $ \ \eta \leq 2^{-160}$.

\medskip
\noindent {\bf (Case 2)  $\eta \leq  \lambda \leq 1-\sqrt{\eta}$ : } We  show that
\begin{equation} \label{allclose-2}
  1 + \sum_{i=1}^N s_i t_j \langle \u_i,\v_j\rangle^{8}  \geq
     \sum_{i=1}^N s_i r_j \langle \u_i,\w_j\rangle^{8}  +
      t_j r_j (1-\eta)^{8}  + (2\eta-\eta^2)^4.
\end{equation}
\noindent {\bf (Subcase i) $t_j \not= r_j$ : } In this case it suffices to show that
 $$   1 + (1-\eta)^{8} \geq  \sum_{i=1}^N  \langle \u_i,\v_j\rangle^{8}   +
     \sum_{i=1}^N   \langle \u_i,\w_j\rangle^{8}  +
          (2\eta-\eta^2)^4.   $$
Again, as before, we have that for every $1\leq i \leq N,$  $$|\langle \u_i,\w_j\rangle |\leq \lambda
\leq 1-\sqrt{\eta},$$ and $$|\langle \u_i,\v_j\rangle|\leq \lambda+\eta \leq
 1-\sqrt{\eta}+\eta.$$ Thus, it suffices to prove that
$$   1 + (1-\eta)^{8} \geq  (1-\sqrt{\eta}+\eta)^{6} + (1-\sqrt{\eta})^{6}
                               +     (2\eta-\eta^2)^4.   $$
This also holds when $\ \eta \leq 2^{-160}$.

\medskip
\noindent {\bf (Subcase ii) $t_j = r_j$ : } We need to prove
(\ref{allclose-2}).  It suffices to show that
\begin{equation*} \label{allclose-3}
  1 - (1-\eta)^{8} - (2 \eta-\eta^2)^4  \
 \geq \  \sum_{i=1}^N  |\langle \u_i,\w_j\rangle|^{8} - \langle \u_i,\v_j\rangle^{8} |
  = \sum_{i=1}^N | \theta_i^{8} - \mu_i^{8}|
\end{equation*}
where  $\theta_i \defeq |\langle \u_i,\w_j\rangle|, \ \ \mu_i \defeq |\langle \u_i,\v_j\rangle|$.
Clearly,
$$ |\theta_i - \mu_i| \leq
   | \langle \u_i,\v_j\rangle - \langle \u_i,\w_j \rangle | \leq  1-  \langle\v_i,\w_j \rangle  = \eta. $$
Here, we used the assumption that $(\u_i, \v_j, \w_j)$ satisfy the triangle inequality.
Note also that $\max_{1\leq i \leq N} \theta_i = \lambda $ and $\sum_{i=1}^N
 \theta_i^2 = 1$. Let $J\defeq  \{ i \ | \ \theta_i \leq \eta \}$ and
 $I \defeq \{ i \ | \ \theta_i \geq \eta\}$. We have,
 \begin{eqnarray*}
 \sum_{i=1}^N | \theta_i^{8} - \mu_i^{8}|  & \leq & \sum_{i \in J} (\theta_i^{8} +
\mu_i^{8}) + \sum_{i \in I} ((\theta_i+\eta)^{8} - \theta_i^{8})  \\
 & \leq &  (\eta)^{6} + (2 \eta)^{6} + \sum_{i \in I}
          ((\theta_i+\eta)^{8} - \theta_i^{8}).
\end{eqnarray*}
Lemma \ref{lemma:kth-power}, which appears after this proof, implies that the summation on the  last line above  is bounded by
$$\sum_{l=1}^{6} \binom{8}{l}
\lambda^{6-l}\eta^l +9 \eta^{6}.$$
Thus, it suffices to show that
$$  1 - (1-\eta)^{8} - (2 \eta-\eta^2)^4  \
 \geq \  \sum_{l=1}^{6} \binom{8}{l} \lambda^{6-l}
   \eta^l              + (4\eta)^{6}.$$
This is true if
$$ 8\eta - \sum_{l=2}^{8} \binom{8}{l} \eta^l -  (2\eta-\eta^2)^4
 \geq 8 \lambda^{5} \eta + \sum_{l=2}^{8} \binom{8}{l} \eta^l + (4\eta)^{6}.$$
This is true if
 $ 8\eta (1-\lambda^{5}) \geq \eta^2 (2^{8} + 2^{8} + 1 + 4^{8}).$
This is true if $8\eta\sqrt{\eta} \geq \eta^2 \cdot 4^{9},$ which holds
when $\ \eta \leq 2^{-160}.$ Note that we used the fact that $\lambda \leq 1-\sqrt{\eta}.$

\medskip
\noindent {\bf (Case 3)  $1-\sqrt{\eta} \leq \lambda \leq 1-\eta^2$ : }
We have $\langle \v_j,\w_j\rangle = 1-\eta$, $\langle \u_j,\w_j\rangle =\lambda =: 1-\zeta.$
This implies that $\langle \u_j,\v_j\rangle = 1-\delta,$ where by the triangle inequality
$$ \eta \leq \zeta + \delta, \ \ \delta \leq \eta + \zeta, \ \ \zeta \leq
   \eta + \delta. $$
Thus, to prove (\ref{allclose-1.5}), it suffices to show that
\begin{multline*}
 1 + s_j t_j \langle \u_j,\v_j\rangle^{8} \geq
        s_j r_j \langle \u_j,\w_j\rangle^{8} + t_j r_j \langle \v_j,\w_j\rangle^{8}
      +  (2\eta-\eta^2)^s + (2\zeta-\zeta^2)^4 +
        (2\delta-\delta^2)^4.
 \end{multline*}
Depending on signs $s_j, t_j, r_j$, this reduces to proving one of
the three cases:
$$ 1 + (1-\delta)^{8}  \geq
        (1-\zeta)^{8} + (1-\eta)^{8}   +  (2\eta-\eta^2)^4 + (2\zeta-\zeta^2)^4 +
        (2\delta-\delta^2)^4. $$
$$ 1 + (1-\eta)^{8}  \geq
        (1-\zeta)^{8} + (1-\delta)^{8}  + (2\eta-\eta^2)^4 +
 (2\zeta-\zeta^2)^4 +
        (2\delta-\delta^2)^4.   $$
 $$ 1 + (1-\zeta)^{8}  \geq
        (1-\eta)^{8} + (1-\delta)^{8}  + (2\eta-\eta^2)^4 +
 (2\zeta-\zeta^2)^4 +        (2\delta-\delta^2)^4.   $$
We  prove the first case, and the remaining two are proved in a
similar fashion. We have that 
\begin{eqnarray*}
1+(1-\delta)^{8} - (1-\zeta)^{8}- (1-\eta)^{8}
  & \geq  & 1+(1-(\zeta+\eta))^{8} - (1-\zeta)^{8} - (1-\eta)^{8}  \\
 & \geq &   8\cdot 7 \cdot \zeta\eta - \sum_{\begin{subarray}{c} 3 \leq i+j \leq 8 \\ i \geq 1, j \geq 1 \end{subarray}}
    \binom{8}{i+j}\binom{i+j}{i} \zeta^i \eta^{j} \\
 & \geq & 8 \cdot 7 \cdot \zeta \eta - \ 2^{32} \zeta \eta \cdot \max \{\zeta, \eta, \delta\} \\
&  \geq & \min \{\zeta \eta, \eta\delta, \zeta\delta\},
\end{eqnarray*}
provided that $2^{32} \max\{\zeta, \eta, \delta\} \leq 1.$
Thus, it suffices to have
  $$   \min\{\zeta \eta, \eta\delta, \zeta\delta\}
       \geq (2\eta-\eta^2)^4 + (2\zeta-\zeta^2)^4 +
        (2\delta-\delta^2)^4.    $$
This is clearly true if $\zeta, \eta, \delta$ are within a
quadratic factor of each other, and $\eta \leq 2^{-160}$. On the contrary if $\delta < \eta^2,$ since we already have $\delta \leq \eta + \zeta$
from the triangle inequality, it reduces to Case (2) by setting $\eta$ to
$\delta$ and setting $\lambda$ to $1 - \eta.$

\medskip
\noindent {\bf (Case 4) $1-\eta^2 \leq \lambda$ : } This is essentially same
as Case (2). Just interchange $1-\eta$ with $\lambda$ and interchange $\u_i, \v_i$ for
every $i$. This completes the proof of the  lemma.
\end{proof}

\begin{lemma} \label{lemma:kth-power}
Let $\eta, \lambda$  and
$\{\theta_i\}_{i=1}^N$ be non-negative reals, such that
 $\sum_{i=1}^N \theta_i^2 \leq 1,$ and for all $i,$ $\eta \leq
 \theta_i \leq \lambda.$
Then
  $$ \sum_{i=1}^N   ((\theta_i + \eta )^{8}- \theta_i^{8})  \leq
\sum_{l=1}^{6} \binom{8}{l} \lambda^{6-l} \eta^l +
            9 \eta^{6}. $$
\end{lemma}
\begin{proof} Clearly, $N \leq \nfrac{1}{\eta^2}$.
\begin{eqnarray*}
\sum_{i=1}^N (\theta_i+\eta)^{8} - \theta_i^{8}
 & = & \sum_{i=1}^N \sum_{l=1}^{8}  \binom{8}{l}
 \theta_i^{8-l}\eta^l  \\
  & = & \sum_{l=1}^{8-2} \binom{8}{l}
  \sum_{i=1}^N \theta_i^{8-l} \eta^l  + 8\cdot \left(\sum_{i=1}^N
  \theta_i\right) \eta^{7} + N \eta^{8} \\
  & \leq & \sum_{l=1}^{6} \binom{8}{l} \lambda^{6-l}
   \eta^l +
  8\cdot \sqrt{N} \eta^{7} + N \eta^{8} \\
   & \leq & \sum_{l=1}^{6} \binom{8}{l} \lambda^{6-l}
   \eta^l +
            9 \eta^{6}.
\end{eqnarray*}
\end{proof}

\end{document}